\documentclass[journal]{IEEEtran}

\newcommand{\avec}{{\bf{a}}}

\newcommand{\bvec}{{\bf{b}}}
\newcommand{\cvec}{{\bf{c}}}
\newcommand{\dvec}{{\bf{d}}}

\newcommand{\yvec}{{\bf{y}}}

\newcommand{\wvec}{{\bf{w}}}
\newcommand{\xvec}{{\bf{x}}}

\newcommand{\vvec}{{\bf{v}}}
\newcommand{\gvec}{{\bf{g}}}

\newcommand{\onevec}{{\bf{1}}}
\newcommand{\zerovec}{{\bf{0}}}

\newcommand{\alphavec}{{\bf{\alpha}}}

\newcommand{\Phivec}{{\bf{\Phi}}}
\newcommand{\Lambdamat}{{\bf{\Lambda}}}

\newcommand{\Gammamat}{{\bf{\Gamma}}}
\newcommand{\Amat}{{\bf{A}}}
\newcommand{\Bmat}{{\bf{B}}}
\newcommand{\Cmat}{{\bf{C}}}
\newcommand{\Dmat}{{\bf{D}}}

\newcommand{\Imat}{{\bf{I}}}
\newcommand{\Lmat}{{\bf{L}}}

\newcommand{\Mmat}{{\bf{M}}}

\newcommand{\Vmat}{{\bf{V}}}
\newcommand{\Wmat}{{\bf{W}}}

\newcommand{\Sigmamat}{{\bf{\Sigma}}}

\newcommand{\define}{\stackrel{\triangle}{=}}

\def\alphavec{{\mbox{\boldmath $\alpha$}}}

\def\muvec{{\mbox{\boldmath $\mu$}}}

\def\alphavecsmall{{\mbox{\boldmath {\scriptsize $\alpha$}}}}

\newcommand{\be}{\begin{equation}}
\newcommand{\ee}{\end{equation}}
\newcommand{\beqna}{\begin{eqnarray}}
\newcommand{\eeqna}{\end{eqnarray}}

\usepackage{amsfonts,amsthm,amsmath,amssymb}
\usepackage{graphicx,subcaption}
\usepackage[noadjust]{cite}
\usepackage[ruled,vlined]{algorithm2e}
\usepackage[dvipsnames]{xcolor} 
\DeclareMathAlphabet{\pazocal}{OMS}{zplm}{m}{n}
\DeclareMathOperator*{\argmin}{argmin}
\DeclareMathOperator*{\argmax}{argmax}
\DeclareMathOperator{\EX}{\mathbb{E}}
\newtheorem{Claim}{Claim}

\newtheorem{Theorem}{Theorem}

\begin{document}
\title{Bayesian Estimation of Graph Signals}
	\author{Ariel Kroizer,
	  Tirza Routtenberg, \IEEEmembership{Senior Member, IEEE}, and
	    Yonina C. Eldar,~\IEEEmembership{Fellow, IEEE}
		\thanks{©2022 IEEE.  Personal use of this material is permitted.  Permission from IEEE must be obtained for all other uses, in any current or future media, including reprinting/republishing this material for advertising or promotional purposes, creating new collective works, for resale or redistribution to servers or lists, or reuse of any copyrighted component of this work in other works.\\
		A. Kroizer and T. Routtenberg are with the School of Electrical and Computer Engineering, Ben-Gurion University of the Negev, Beer Sheva, Israel. 
		Y. C. Eldar is with the Faculty of Mathematics and Computer Science, Weizmann Institute of Science, Rehovot, Israel.
e-mail: arielkro@post.bgu.ac.il, tirzar@bgu.ac.il, yonina@weizmann.ac.il. This work is partially supported by the  Israeli Ministry of National Infrastructure, Energy and Water Resources.}
		}

	\maketitle

\begin{abstract}
We consider the problem of recovering random graph signals from nonlinear measurements. 
For this setting, closed-form Bayesian estimators are usually intractable and even numerical evaluation may be difficult to compute for large networks. 
In this paper, we propose a graph signal processing (GSP) framework for random graph signal recovery that utilizes information on the structure behind the data. First, we develop the GSP-linear minimum mean-squared-error (GSP-LMMSE) estimator, which minimizes the mean-squared-error (MSE) among estimators that are represented as an output of a graph filter.
The GSP-LMMSE estimator is based on {\em{diagonal}} covariance matrices in the graph frequency domain, and thus, has reduced complexity compared with the LMMSE estimator. This property is especially important when using the sample-mean versions of these estimators that are based on a training dataset.
We then state conditions under which the low-complexity GSP-LMMSE estimator coincides with the optimal LMMSE estimator. 
Next, we develop an approximate parametrization of the GSP-LMMSE estimator by shift-invariant graph filters by solving a weighted least-squares (WLS) problem.
We present three implementations of the parametric GSP-LMMSE estimator for typical graph filters. 
These parametric graph filters  are more robust to outliers and to network topology changes.
In our simulations, we evaluate the performance of the proposed GSP-LMMSE estimators for the problem of state estimation in power systems, which can be interpreted as a graph signal recovery task.
We show that the proposed sample-GSP estimators outperform the sample-LMMSE estimator for a limited training dataset and that the parametric GSP-LMMSE estimators are more robust to topology changes in the form of adding/removing vertices/edges.
\end{abstract}

\section{Introduction} Graph signals arise in various applications such as the study of brain signals \cite{7544580} and sensor networks \cite{1583238}. The area of graph signal processing (GSP) has gained considerable interest in the last decade. GSP theory extends concepts and techniques from traditional digital signal processing (DSP) to data indexed by generic graphs, including the graph Fourier transform (GFT), graph filter design \cite{8347162,Shuman_Ortega_2013,Isufi_Leus2017}, and sampling and recovery of graph signals \cite{9244650,6854325,7352352,9043719}.
Many modern network applications involve complex models and large datasets and are characterized by nonlinear models \cite{9343697,8347160}, for example, the brain network connectivity \cite{shen2016nonlinear}, environment monitoring \cite{8496842},
and power flow equations in power systems \cite{drayer2018detection,Grotas_2019,shaked2021identification}.
The recovery of graph signals in such networks is often intractable, especially for large networks.
For example, the recovery of voltages from power measurements is an NP-hard nonconvex optimization problem \cite{bienstock2019strong}, which is at the core of power system analysis \cite{Abor}. In this case, the graph represents an electrical network, and the signals are the voltages.
Thus, the development of GSP methods for the estimation of graph signals has significant practical importance, 
in addition to its contribution to the enrichment of theoretical statistical GSP tools.
 
The recovery of random graph signals can be performed by state-of-the-art Bayesian estimators, such as the minimum mean-squared-error (MMSE) and the linear MMSE (LMMSE) estimators.
However, MMSE estimation is often computationally intractable and does not have a closed-form expression in nonlinear models.
The LMMSE can be used when the second-order statistics are completely specified. 
In some cases (e.g. \cite{6891254,7931630,berman2020resource}), accurate characterization of the nonlinear model is possible by using tools such as Bussgang's theorem.
However, for the general case, the distributions (e.g. the covariance matrices) of the desired graph signal and the observations are difficult to determine.
Moreover, in many practical applications, the graph signal has a broad correlation function so that estimating this correlation from data with high accuracy often necessitates a larger sample size than is available \cite{Eldar_merhav} and requires stationarity of the signals.
Low-complexity estimation algorithms have been considered as an alternative. For example, in \cite{Edfors_Borjesson_1996}, a low-rank approximation is applied to the LMMSE estimator by using the singular value decomposition (SVD) of the covariance matrices. 
The dual-diagonal LMMSE (DD-LMMSE) channel estimation algorithm, which is based on the diagonal of the covariance matrices, was proposed in \cite{dd}.
However, these methods may lead to considerable performance loss compared with LMMSE estimation.

Graph filters have been used for many signal processing tasks, such as denoising \cite{7032244,ZHANG20083328}, classification \cite{6778068}, and anomaly detection \cite{drayer2018detection}. 
The design of graph filters to obtain a desired graph frequency response for the general case has been studied and analyzed in various works \cite{6808520,sparse_paper,Isufi_Leus2017,7926424}.
Model-based recovery of a graph signal from noisy measurements by graph filters for {\em{linear}} models was treated in \cite{7117446,7032244,Isufi_Leus2017,7891646}.
To derive classical graph filters, such as the Wiener filter, more restrictive assumptions on the graph signal are required  \cite{7891646}.
Nonlinear graph filters were considered in \cite{8496842}, but they require higher-order statistics that are not completely specified in the general case.
Graph neural network approaches were considered in \cite{Isufi_Ribeiro,gama2020graphs}; however, data-based methods necessitate extensive training sets, and result in nonlinear estimators, while in this paper, we focus on linear estimation with limited training data.

Fitting graph-based models to given data was considered in \cite{7763882,Hua_Sayed_2020,rey2021robust}. 
In \cite{confPaper}, we proposed a two-stage method for  graph signal estimation from a known nonlinear observation model, which is based on fitting a 
 graph-based model and then implementing least-squares recovery on the approximate model.
However, model-fitting approaches aim to minimize the modeling error and in general have significantly lower performance  than  estimators that minimize the estimation error directly.
A fundamental question remains regarding how to use knowledge of the nonlinear physical model and on the graph topology to obtain a low-complexity estimator for the general nonlinear case that has optimal MSE performance.

In this paper, we consider the estimation of random graph signals from a nonlinear observation model. First, we present the sample-LMMSE estimator, in which the analytical expressions of the LMMSE estimator are replaced by their estimated values; this requires use of large training datasets.
Next, we propose the GSP-LMMSE estimator, which minimizes the MSE among the subset of estimators that are represented as an output of a graph filter.
We discuss the advantages of the GSP-LMMSE estimator in terms of complexity and show that 1) for models with diagonal covariance matrices in the graph frequency domain,  the GSP-LMMSE and LMMSE estimators  coincide; and
2) for linear models,
the GSP-LMMSE estimator coincides with the graphical Wiener filter \cite{7891646}.
We develop the MSE-optimal parametrization of the GSP-LMMSE estimator by shift-invariant graph filters that approximate the graph frequency response of the GSP-LMMSE estimator.
The parameterized estimators can be applied to any topology, making them more robust to network topology changes.
We also implement three types of the parameterized GSP-LMMSE estimator based on well-known graph filters.
Finally, we perform numerical simulations for the problem of state estimation in power systems. 
We show that in this case the sample-GSP estimators outperform the sample-LMMSE estimator for a limited training dataset and coincide with the sample-LMMSE estimator otherwise. Moreover, the proposed estimators are more robust to changes in the network topology in the form of adding/removing vertices/edges.

The rest of this paper is organized as follows. In Section \ref{background_sec} we introduce the basics of GSP and three examples of graph filters.
In Section \ref{problem_formulation}, we formulate the estimation problem and present the MMSE and the LMMSE estimators.
In Section \ref{GSP_estimator_section}, we develop the proposed GSP-LMMSE estimator and present parameterizations of the GSP-LMMSE estimator in Section \ref{General_design}.
Simulation are shown in Section \ref{simulation}.
Finally, the paper concluded in Section \ref{conclusion}.

We denote vectors by boldface lowercase letters and matrices by boldface uppercase letters.
The operators $(\cdot)^T$, $(\cdot)^{-1}$, and $(\cdot)^\dagger$ represent the transpose, inverse, and pseudo-inverse, respectively. The notation 
$\circ$ denotes the Hadamard product.
For a matrix $\Amat$, $\text{rank}(\Amat)$ is its rank.
For a vector $\avec$, ${\text{diag}}(\avec)$ is a diagonal matrix whose $i$th diagonal entry is $\avec_i$; when applied to a matrix, ${\text{diag}}(\Amat)$ is a vector collecting the diagonal elements of $\Amat$.
The $m$th element of the vector $\avec$ is written as $a_m$ or $[\avec]_m$. The $(m,q)$th element of the matrix $\Amat$ is written as $A_{m,q}$ or $[\Amat]_{m,q}$.
The identity matrix of dimension $N$ is written as $\Imat_N$ and the vector $\zerovec_N$ is a length $N$ vector of all zeros.
The multivariate Gaussian distribution of $\yvec$ with mean vector, $\muvec$, and covariance matrix, $\Sigmamat$, is denoted by $\yvec \sim \pazocal{N}(\muvec,\Sigmamat)$.
The cross-covariance matrix of the vectors $\avec$ and $\bvec$ is denoted by $\Cmat_{\avec\bvec}\triangleq \EX[(\avec-\EX[\avec]) (\bvec-\EX[\bvec])^T]$.

\section{Background: Graph Signal Processing (GSP)}
\label{background_sec}
We begin by reviewing GSP and general graph filters in Subsection \ref{Background}. Three commonly-used graph filters are then presented in Subsections \ref{LFI_representation}, \ref{ARMA_representation}, \ref{LR_representation}, and will be used later in the paper.
\subsection{GSP background}\label{Background}
Consider an undirected, connected, weighted graph ${\pazocal{G}}({\pazocal{V}},\xi,\Wmat)$, where $\pazocal{V}$ and $\xi$ are sets of vertices and edges, respectively. 
The matrix $\Wmat \in \mathbb{R}^{N \times N}$ is the non-negative weighted adjacency matrix of the graph, where $N \define |\pazocal{V}|$ is the number of vertices in the graph. If there is an edge $(i, j) \in \xi$ connecting vertices $i$ and $j$, the entry $\Wmat_{i,j}$ represents the weight of the edge; otherwise, $\Wmat_{i,j} = 0$. A common way to represent the graph topology is by the Laplacian matrix, which is defined by
\begin{equation}
    \Lmat \define \text{diag}\left(\Wmat\onevec\right) - \Wmat.
\end{equation}
The Laplacian matrix, $\Lmat$, is a real and positive semidefinite matrix and its  eigenvalue decomposition (EVD) is given by
 \begin{equation}
\Lmat=\Vmat\Lambdamat\Vmat^{T},   \label{SVD_new_eq}
 \end{equation}
where $\Lambdamat$ is a diagonal matrix consisting of the eigenvalues of 
$\Lmat$, $0= \lambda_1 < \lambda_2 \leq \ldots \leq \lambda_N $, $\Vmat$ is a matrix whose $n$th column, $\vvec_N$, is the eigenvector of $\Lmat$ that is associated with $\lambda_n$, and $\Vmat^{T}=\Vmat^{-1}$.
We assume, without loss of generality, that ${\pazocal{G}}$ is a connected graph, i.e., that $\lambda_2\neq0$ \cite{Newman_2010}.

In this paper, a {\em{graph signal}} is 
an $N$-dimensional vector, $\avec$, that assigns a scalar value to each vertex, i.e., each entry $a_n$ denotes the signal value at vertix $n$, for $n=1,\ldots, N$.
The GFT of the graph signal $\avec$ is defined as \cite{Shuman_Ortega_2013}
 \begin{equation}
\label{GFT}
\tilde{\avec} \triangleq \Vmat^{T}\avec. 
 \end{equation}
Similarly, the inverse GFT (IGFT) of $\tilde{\avec}$ is given by $\Vmat\tilde{\avec}$.
Finally, a graph signal is a graph-bandlimited signal with cutoff graph frequency $N_s$ if it satisfies \cite{8347162}
\begin{equation} \label{bandlimited_def}
    \tilde{a}_n =0,~ n =N_s+1,\dots,N.
\end{equation}

 Graph filters are useful tools for various GSP tasks. 
Linear and shift-invariant graph filters with respect to (w.r.t.) the graph shift operator (GSO) play essential roles in GSP. 
These filters generalize linear time-invariant filters used in DSP for time series, and enable performing tractable operations over graphs \cite{8347162,Shuman_Ortega_2013}.
A graph filter is a function $f(\cdot)$ applied to a GSO, where here we use the GSO given by $\Lmat$,
that allows an eigendecomposition as follows \cite{8347162}:
 \begin{equation} \label{laplacian_graph_filter}
  	f(\Lmat)= \Vmat f(\Lambdamat)\Vmat^T,
  \end{equation}
 where $f(\Lambdamat)$ is a diagonal matrix.
That is, $f(\lambda_n)$ is the graph frequency response of the filter at graph frequency $\lambda_n$, $n=1,\dots,N$, and $f(\Lmat)$ is  diagonalized by the eigenvector matrix $\Vmat$ of  $\Lmat$.
We assume that the graph filter, $f(\cdot)$, is a well-defined function on the spectrum $\{\lambda_1,\ldots,\lambda_N\}$ of $\Lmat$.

Throughout this paper, we consider three commonly-used parametrizations of the graph filter function, $f(\cdot)$, 
that are appropriate for modeling low-pass graph filters.  For clarity of representation,
 we will write this specific function as $h(\cdot;\alphavec)$, where $\alphavec$ contains the graph filter parameters.
 It should be noted that 
 under simple conditions, all  filters in the form of \eqref{laplacian_graph_filter}  can be represented as a finite polynomial of $\Lmat$ \cite{6638849}.
 Linear graph filters can be implemented locally, e.g., with exchanges of information
among neighbors \cite{6638849}.
 However, due to the nature of the polynomial fitting problem,  linear graph filters usually have limited accuracy when used for approximating a desired graph frequency response (see, e.g.  in \cite{sparse_paper,6808520,7926424}). This is further discussed  in Subsection \ref{GSP_estimator_section},  where the benefits of the presented parametrization in this context are outlined.

\subsection{Linear pseudo-inverse graph filter} \label{LFI_representation}
In a linear graph filter, the filter is a polynomial of a GSO, such as the Laplacian matrix \cite{8347162,Shuman_Ortega_2013}.
In addition,
the Moore-Penrose pseudo-inverse of the Laplacian matrix plays an essential role in many graph-based applications \cite{4072747,4053117}.
Inspired by the linear graph filter,
we define the linear pseudo-inverse graph filter,
$f(\Lmat)=h^{{\text{LPI}}}(\Lmat^\dagger;\alphavec^{\text{LPI}})$, as
\begin{equation}
\label{pseudo_inverse_filter_def}
h^{\text{LPI}}(\Lmat^\dagger;\alphavec^{\text{LPI}})  \define h_{0}\Imat_N+h_{1}\Lmat^\dagger +\ldots+h_{K}(\Lmat^\dagger)^{K},
\end{equation}
where $K$ is the filter order and the filter coefficients vector is 
\begin{equation} \label{LPI_coefficients}
    \alphavec^{\text{LPI}} = 
    \begin{bmatrix}
        h_0,\dots,h_K
    \end{bmatrix}^T
    \in \mathbb{R}^{K+1},
\end{equation}
with $K<N$ due to the Cayley-Hamilton theorem.
From \eqref{laplacian_graph_filter} and \eqref{pseudo_inverse_filter_def}, we conclude that 
the frequency response of the linear pseudo-inverse graph filter at graph frequency $\lambda_n$ can be expressed as
\begin{equation}\label{LPI graph filter}
h^{{\text{LPI}}}(\lambda_n;\alphavec^{\text{LPI}})= 
 \begin{cases}
h_0, &n=1\\
 \sum_{k=0}^K h_k\lambda_n^{-k},  & n=2,\dots,N.
 \end{cases}
 \end{equation}

\subsection{Autoregressive Moving Average (ARMA) graph filter} \label{ARMA_representation} 
Similar to temporal ARMA filters \cite{StatisticalDigitalSignalProcessingHayes}, an
ARMA graph filter is characterized by a rational polynomial in
the Laplacian matrix \cite{Isufi_Leus2017}.
In this case, $f(\Lmat)=h^{{\text{ARMA}}}(\Lmat;\alphavec^{{\text{ARMA}}})$, where the ARMA graph filter is defined as 
\begin{equation} \label{eq:0}
	h^{{\text{ARMA}}}(\Lmat;\alphavec^{{\text{ARMA}}}) \define \left( \Imat_N+\sum_{r=1}^{R} a_{r}\Lmat^{r} \right)^{-1}\sum_{q=0}^Q c_q\Lmat^q
\end{equation}
and the ARMA filter coefficient vector is
\begin{equation} \label{ARMA_coefficients}
    \alphavec^{{\text{ARMA}}} = 
   [\avec^T,
        \cvec^T]^T,
\end{equation}
where $\avec= [1,a_1,\ldots,a_{R}]^T$ and 
$\cvec= [c_0,\ldots,c_Q]^T$.
It is assumed that $\Imat_N+\sum_{r=1}^{R} a_{r}\Lmat^{r}$ is a non-singular matrix.
From \eqref{laplacian_graph_filter} and \eqref{eq:0}, the graph-ARMA filter frequency response at graph frequency $\lambda_n$ can be expressed as
\begin{equation}\label{ARMA graph filter}
 h^{{\text{ARMA}}}(\lambda_n;\alphavec^{{\text{ARMA}}})= \frac{\sum_{q=0}^Qc_q\lambda_n^q}{1 +\sum_{r=1}^{R} a_{r}\lambda_n^{r}},
 \end{equation}
  $ n =1,\dots,N$, where we assume that $1+ \sum_{r=1}^{R}a_{r}\lambda_n^{r} \neq 0$.

The linear graph filter, which is a polynomial of the Laplacian matrix, $\Lmat$:
\begin{equation}
\label{Linear_filter_def}
    h^{\text{lin}}(\Lmat;\alphavec^{\text{lin}})  \define h_{0}\Imat_N+h_{1}\Lmat +\ldots+h_{K}\Lmat^{K},
\end{equation}
where $K$ is the filter order and 
$
    \alphavec^{\text{lin}} = 
    \begin{bmatrix}
        h_0,\dots,h_K
    \end{bmatrix}^T$ is the filter coefficients vector,
can be obtained as a special case of the ARMA graph filter by setting $a_{r}=0$, $r=1,\ldots,R$  and $Q=K$ in \eqref{eq:0}.
However, the ARMA graph filter improves the accuracy of the approximation of the desired graph frequency response and requires fewer coefficients compared to linear graph filters \cite{sparse_paper}.
In addition, unlike the linear pseudo-inverse graph filter, the ARMA graph filter has  distributed implementations, as shown in \cite{sparse_paper,7001054}.

\subsection{Low-rank graph filter}\label{LR_representation}
In many GSP applications, the graph signals are assumed to be {\em{bandlimited}} in the graph spectrum domain \cite{7208894}.
In particular, under the assumption of a  low-frequency graph signal,
 we are interested in filtering only the lowest graph frequencies and equating the output graph signal at high graph frequencies to zero. To this end, we use a low-rank graph filter.

A low-rank graph filter can be designed based on any graph filter, $f(\cdot)$, by substituting the following low-rank matrix:
\begin{equation} \label{L_bar_def}
    \bar{\Lmat} = \sum\nolimits_{n=1}^{N_s} \lambda_n\vvec_n\vvec_n^T,
\end{equation}
which contains the $N_s$ smallest eigenvalues of the Laplacian matrix, $\Lmat$, and their associated eigenvectors, in $f(\cdot)$. That is, we use the filter $f(\bar{\Lmat})$ instead of $f(\Lmat)$. Note that, $\text{rank}(\bar{\Lmat})=N_s-1<N-1$ and that $\bar{\Lmat}$ is not a Laplacian matrix. 
For these filters, we use the convention that the zero power of the matrix $\bar{\Lmat}$ is
\begin{equation} \label{I_bar_def}
    \bar{\Lmat}^0 = \sum_{n=1}^{N_s} \vvec_n\vvec_n^T.
\end{equation}
Substituting \eqref{L_bar_def} in \eqref{laplacian_graph_filter}, the frequency response of a low-rank graph filter at the graph frequency $\lambda_n$ is
\begin{equation}\label{LR graph filter}
 [f(\Lambdamat)]_{n,n}= 
 \begin{cases}
 f(\lambda_n),  &n=1,\dots,N_s\\
0, &n =N_s+1,\dots, N.
 \end{cases}
 \end{equation}
The advantage of the low-rank graph filter is that it requires less filter coefficients and computation of only the $N_s$ smallest eigenvalues and eigenvectors.
As a result, the evaluation of the filter coefficients has lower computational complexity. 

In this paper, we use the low-rank ARMA (LR-ARMA) graph filter,
\begin{equation} \label{LR_ARMA_def}
    f(\bar{\Lmat})=h^{{\text{ARMA}}}(\bar{\Lmat};{\alphavec}^{{\text{LR}}}),
\end{equation} 
where $h^{{\text{ARMA}}}(\cdot;\cdot)$ is defined in \eqref{eq:0} and the reduced, low-rank Laplacian matrix is defined in \eqref{L_bar_def}.
The filter coefficients vector in this case is given by
\begin{equation} \label{LR_ARMA_coefficients}
\begin{split}
    \alphavec^{{\text{LR}}} = 
[
  (\avec^{\text{LR}})^T,
        (\cvec^{\text{LR}})^T]^T,
    \end{split}
\end{equation}
where $\avec^{\text{LR}}= [
  1,a_1^{\text{LR}},\ldots,a_{R}^{\text{LR}}]^T$ and $\cvec^{\text{LR}}=[c_0^{\text{LR}},\ldots,c_Q^{\text{LR}}]^T$.
From \eqref{laplacian_graph_filter}, \eqref{eq:0}, \eqref{ARMA graph filter}, and \eqref{LR graph filter}, the graph-LR-ARMA filter frequency response at graph frequency $\lambda_n$ can be expressed as
\begin{equation} \label{LR_ARMA_frequency_domain}
\begin{split}
h^{{\text{LR-ARMA}}}(\lambda_n;\alphavec^{\text{LR}})
  =\begin{cases}
 \frac{\sum_{q=0}^{Q_{{\text{LR}}}} c_q^{\text{LR}}\lambda_n^q}{1+\sum_{r=1}^{{R}_{{\text{LR}}}}a_{r}^{\text{LR}}\lambda_n^{r}}, & n =1,\dots,N_s\\
 0, &n =N_s+1,\dots, N,
 \end{cases}
 \end{split}
 \end{equation}
 where we assume that $1+ \sum_{r=1}^{R}a_{r}^{\text{LR}}\lambda_n^{r} \neq 0$, $\forall n = 1,\dots,N_s$.
However, it should be noted that 
 distributed implementations of the   graph-LR-ARMA  filter do not exist. Moreover,
the  LR-ARMA filter may not be a suitable parametrization of a graph frequency response that contains nonzero (small) values at high graph frequencies.

\section{Model and problem formulation} \label{problem_formulation}
We now introduce the problem of estimating a random graph signal by observing its noisy nonlinear function, which is determined by the graph topology.
In Subsection \ref{model_sub_sec}, we present the measurement model.
In Subsection \ref{Monte_Carlo_LMMSE_estimator}, the LMMSE and the sample-LMMSE estimators are derived.

\subsection{Model and problem formulation}\label{model_sub_sec}
Consider the problem of recovering a random input graph signal, $\xvec\in\mathbb{R}^N$, based on the following nonlinear measurement model: 
\begin{equation} \label{Model}
	\yvec = \gvec(\Lmat,\xvec) + \wvec,
\end{equation}
where the measurement function, $\gvec:\mathbb{R}^{N \times N} \times{\mathbb{R}}^N \rightarrow {\mathbb{R}}^N$, and the Laplacian matrix, $\Lmat$, which represents the influence of the graph topology, are assumed to be known. 
The noise, $\wvec$, has a known probability density function (pdf), $f_\wvec$, with zero mean and covariance matrix $\Cmat_{\wvec\wvec}$. We assume that the pdf of $\xvec$, $f_\xvec$, is known, and that $\wvec$ and $\xvec$ are  independent.

This model comprises a broad family of statistical estimation problems over graphs. For example, a core problem of power system analysis is the recovery of voltages from power measurements \cite{Abor}, where both the voltages and the powers can be considered as graph signals \cite{routtenberg2020,drayer2018detection,dabush2021state} and the Laplacian matrix is the susceptance matrix. This problem is NP-hard and nonconvex \cite{bienstock2019strong}; It will be discussed further in Section \ref{simulation}.
Similarly, for a water distribution system we are interested in estimating the nodal demands from measurements of the pressure heads and pipe flows using the nonlinear relation between them \cite{letting2017estimation}.
Another example is the problem of estimating a  graph signal 
based
on noisy observations  \cite{alaoui2019computational}. This setting
comprises a broad family of  estimation problems, including group synchronization on
graphs, community detection, and low-rank matrix estimation.
Finally, this model can be extended to other topology matrices, such as the adjacency matrix and the Markov matrix \cite{Newman_2010}.

We are interested in recovering the input graph signal, $\xvec$, from the observation vector $\yvec$ in \eqref{Model}, based on minimizing the MSE. Thus, we seek an estimator of $\xvec$, as follows:
\begin{equation} \label{Distance}
\hat\xvec^{\text{MMSE}} = \argmin_{\hat\xvec(\yvec,\Lmat) \in \mathbb{R}^N}\EX\left[(\hat\xvec(\yvec,\Lmat) - \xvec )^T(\hat\xvec(\yvec,\Lmat) - \xvec)\right].
\end{equation} 
In the general case where the statistics of $\xvec$ and $\yvec$ are known, the MMSE estimator, $\hat\xvec^{\text{MMSE}} = \EX[\xvec|\yvec]$, solves \eqref{Distance}. Computing $\hat\xvec^{\text{MMSE}}$ requires an expression for the posterior pdf of $\xvec$ given $\yvec$, denoted by $f_{\xvec|\yvec}$.
For the considered model, the pdf of $\yvec$ satisfies
\begin{equation}
\label{pdf_x}
	f_{\yvec} = f_{\gvec(\Lmat,\xvec )} \ast f_{\wvec},
\end{equation}
where $\ast $ denotes the convolution operator and $f_{\gvec(\Lmat,\xvec )}$ is the pdf of $\gvec(\Lmat,\xvec )$.
Since $\gvec(\Lmat,\xvec)$ is a nonlinear  function, the pdf of the transformation of $\xvec$ given by $\gvec(\Lmat,\xvec)$ does not have a closed-form expression, and therefore, the pdf of $\yvec$ and the posterior pdf of $\xvec$ given $\yvec$ (i.e., $f_\yvec$ and $f_{\xvec|\yvec}$) do not have analytical closed-form expressions.
As a result, the MMSE estimator does not have a closed-form solution.
Moreover, the computational complexity of the numerical evaluation of the MMSE estimator by multidimensional integration is very high, making the evaluation impractical, especially, for large networks, i.e., large $N$.

\subsection{LMMSE and sample-LMMSE estimators} \label{Monte_Carlo_LMMSE_estimator}
A common sub-optimal approach  is to choose to retain the MMSE criterion, but constrain the subset of estimators.
The LMMSE estimator is the optimal solution to \eqref{Distance} over the subset of estimators that are linear functions of the measurements, $\yvec$, and is given by \cite{Kayestimation}
\begin{equation} \label{LMMSE}
	\hat\xvec^{(\text{LMMSE})} = \EX[\xvec] + \Cmat_{\xvec\yvec}\Cmat_{\yvec\yvec}^{-1}(\yvec - \EX[\yvec]),
\end{equation}
where it is assumed that $\Cmat_{\yvec\yvec}$ is a non-singular matrix. 
The LMMSE estimator can be used when the second-order statistics of $\xvec$ and $\yvec$ are known.
However, in the considered case, the LMMSE estimator is often intractable as well, as an expression for the covariance is not generally known.

Since the pdf of $\xvec$, the statistics of the noise, the Laplacian matrix $\Lmat$, and the measurement function are all known, the {\em{sample-LMMSE}} estimator can be evaluated based on a two-stage procedure
(\cite{6827237}, p. 728 in \cite{van2004optimum}). First, the mean and covariance matrices from \eqref{LMMSE}, $\EX[\yvec]$, $\Cmat_{\xvec\yvec}$, and $\Cmat_{\yvec\yvec}$, are estimated by the sample-means and sample-covariance matrices from a training set\footnote{In our model, $\EX[\xvec]$ is known, and thus, is not replaced by its sample mean.}. Second, these estimators are plugged into the LMMSE estimator in \eqref{LMMSE} to obtain the sample-LMMSE estimator (\cite{6827237}, p. 728 in \cite{van2004optimum}). 

In the first stage, $P$ random samples of $\xvec$, $\{\xvec_p\}_{p=1}^P$, are generated, with a pdf $f_\xvec$. This is the training dataset.
The associated output vectors under the generating model in \eqref{Model} are denoted by $\{\yvec_p\}_{p=1}^P$.
Then, assuming zero-mean noise in \eqref{Model} and using the 
function $\gvec(\cdot,\cdot)$, the sample-mean observation vector from \eqref{Model} is 
\begin{equation}
\label{sample_mean}
	\hat{\yvec} = \frac{1}{P} \sum\nolimits_{p=1}^P\gvec(\Lmat,\xvec_p).
\end{equation}
Similarly, since $\xvec$ and $\wvec$ are assumed to be statistically independent, the sample cross-covariance matrix of $\xvec$ and $\yvec$,
and the sample covariance matrix of $\yvec$
are computed by
\begin{equation} \label{hat_C_theta_x}
	\hat{\Cmat}_{\xvec\yvec} = \frac{1}{P} \sum\nolimits_{p=1}^P (\xvec_p - \EX[\xvec])(\yvec_p - \hat{\yvec})^T
\end{equation}
and
\begin{equation} \label{hat_C_y}
	\hat{\Cmat}_{\yvec\yvec} = \frac{1}{P} \sum\nolimits_{p=1}^P (\yvec_p - \hat{\yvec})(\yvec_p - \hat{\yvec})^T + \Cmat_{\wvec\wvec},
\end{equation}
respectively,
where $\Cmat_{\wvec\wvec}$ is assumed to be known. 
In the second stage, the sample-LMMSE estimator is obtained by plugging the sample-mean and sample covariance matrices from \eqref{sample_mean}, \eqref{hat_C_theta_x}, and \eqref{hat_C_y} into \eqref{LMMSE}, which results in 
\begin{equation} \label{Monte_Carlo_LMMSE}
	\hat\xvec^{(\text{sample-LMMSE})} = \EX[\xvec] + \hat{\Cmat}_{\xvec\yvec}\hat{\Cmat}_{\yvec\yvec}^{-1}(\yvec - \hat{\yvec}).
\end{equation}

The sample-LMMSE estimator requires the computation of the inverse sample covariance matrix of $\yvec$, $\hat{\Cmat}_{\yvec\yvec}$, from \eqref{hat_C_y}. Therefore, a drawback of this method is that it requires an extensive dataset for stable estimation of the inverse sample covariance matrix.
However, the dataset is usually limited in practical applications since the function $\gvec(\xvec,\Lmat)$ may change over time due to changes in the topology.
Moreover, numerical evaluation of the LMMSE may discard information about the relationship between the graph signal and its underlying graph structure. Thus, it is less robust to changes in the graph topology, and there is no straightforward methodology to update the estimator to new topology, which may happen when vertices or edges are added or removed from the network.
Finally, the sample-LMMSE estimator ignores the GSP properties and does not exploit additional information on the graph signal, such as smoothness or graph-bandlimitness, that can improve estimation performance.
On the other hand, existing GSP-based approaches have been developed for simple {\em{linear}}  models (see, e.g. \cite{7891646,Isufi_Leus2017,7032244}).

In order to reduce the computational complexity of the sample-LMMSE estimator, the sample-version of the diagonal-LMMSE estimator can be used \cite{dd}. This estimator minimizes the MSE 
among linear estimators where the estimation matrix that multiplies $\yvec$ is restricted to be diagonal.
By plugging in the sample mean and sample covariance values from \eqref{sample_mean}-\eqref{hat_C_y} in the diagonal-LMMSE estimator (see, e.g. Eq. (18) in \cite{dd}),
 the sample-diagonal-LMMSE estimator is given by
\begin{eqnarray} \label{Monte_Carlo_local_LMMSE}
	\hat\xvec^{(\text{sample-D-LMMSE})} = \EX[\xvec]\hspace{4cm}\nonumber\\ + \text{diag}(\text{diag}(\hat{\Cmat}_{\xvec\yvec}))(\text{diag}(\text{diag}(\hat{\Cmat}_{\yvec\yvec})))^{-1}(\yvec - \hat{\yvec}).
\end{eqnarray}
The use of diagonal matrices in \eqref{Monte_Carlo_local_LMMSE} results in lower computational complexity than that of the sample-LMMSE estimator.
However, the performance of \eqref{Monte_Carlo_local_LMMSE} is significantly poorer than the other methods presented in this paper (as shown in \cite{dd} and was confirmed by our simulations).
This motivates us to seek improved solutions that also involve diagonal matrices and, at the same time, use the given graph.

\section{The GSP-LMMSE estimator} \label{GSP_estimator_section}
In this section, we develop the GSP-LMMSE and the sample-GSP-LMMSE estimators in Subsection \ref{prob_form_sec} and discuss their properties in Subsection \ref{prop_subsection}. Conditions under which the proposed GSP-LMMSE estimator is also the LMMSE estimator are developed in Subsection \ref{separately_model_section}.
\subsection{GSP-LMMSE and sample-GSP-LMMSE estimators}
\label{prob_form_sec}
In this subsection we develop a graph-based linear estimator. 
We consider estimators of the form
\begin{equation} \label{opt_estimator}
	\hat{\xvec} = f(\Lmat)\yvec +\bvec = \Vmat f(\Lambdamat )\Vmat^T\yvec +\bvec ,
\end{equation}
where $\bvec$ is a constant vector,
$\Vmat$ and $\Lambdamat$ are the eigenvector-eigenvalue matrices of the Laplacian matrix,
$\Lmat$, as defined in \eqref{SVD_new_eq}, and $f(\cdot)$ is the graph filter defined in \eqref{laplacian_graph_filter}.
The GSP-LMMSE estimator is an estimator which minimizes the MSE in \eqref{Distance} over the subset of GSP estimators in the form of \eqref{opt_estimator}.
By substituting \eqref{opt_estimator}
 in \eqref{Distance}, the general GSP-LMMSE estimator is defined by
\begin{eqnarray}
\label{LMMSE_GSP_o}
\hat\xvec^{(\text{GSP-LMMSE})}=
 \Vmat \hat{f}(\Lambdamat)\Vmat^T\yvec +\hat{\bvec},
\end{eqnarray}
where
\begin{eqnarray}
\label{Distance_GSP}
\{\hat{f}(\Lambdamat), \hat{\bvec}\}=\hspace{-0.15cm} \argmin_{f(\Lambdamat ) \in {\mathcal{D}}_N,\bvec\in{\mathbb{R}}^N}\EX\left[||\Vmat f(\Lambdamat )\Vmat^T\yvec +\bvec - \xvec ||^2\right]
\end{eqnarray} 
and ${\mathcal{D}}_N$ is the set of diagonal matrices of size $N \times N$.
Since $\Vmat$ is a unitary matrix, i.e., $\Vmat\Vmat^T = \Imat_N$, 
the minimization problem \eqref{Distance_GSP} can be rewritten in the graph frequency domain as follows:
\begin{eqnarray}\label{Distance_GSP_frequency_domain}
\{\hat{f}(\Lambdamat), \hat{\bvec}\}
= \argmin_{f(\Lambdamat ) \in {\mathcal{D}}_N,\tilde{\bvec}\in{\mathbb{R}}^N}\EX\left[|| f(\Lambdamat )\tilde{\yvec} +\tilde{\bvec} - \tilde{\xvec} ||^2\right], 
\end{eqnarray} 
where $\tilde\yvec$, $\tilde\xvec$, and $\tilde\bvec$ are the GFT of $\yvec$, $\xvec$, and $\bvec$, respectively, as defined in \eqref{GFT}.

We first consider the solution for $\tilde{\bvec}$:
\begin{equation} \label{Bayesian_minimize_opt}
\hat{\tilde{\bvec}} = \argmin_{\tilde{\bvec}\in \mathbb{R}^N} \EX\left[|| f(\Lambdamat )\tilde{\yvec} +\tilde{\bvec} - \tilde{\xvec} ||^2\right].
\end{equation}
Since \eqref{Bayesian_minimize_opt} is convex, the optimal solution is obtained by equating the derivative w.r.t. $\tilde{\bvec}$ to zero, which results in
\begin{equation} \label{b_evaluate_old}
	\hat{\tilde{\bvec}} = \EX[\tilde{\xvec}] -  f(\Lambdamat)\EX[\tilde{\yvec}].
\end{equation}
This implies
\begin{equation} \label{b_evaluate_old2}
\hat{\bvec}=\Vmat	\hat{\tilde{\bvec}} = \EX[{\xvec}] -  \Vmat f(\Lambdamat)\Vmat^T \EX[{\yvec}].
\end{equation}
Substituting \eqref{b_evaluate_old} into \eqref{Distance_GSP_frequency_domain} we obtain that the graph frequency response of the GSP-LMMSE estimator is given by
\begin{eqnarray}
 \label{derivative_opt_f_only}
  \hat{f}(\Lambdamat)
  =
  \argmin_{
  f(\Lambdamat)\in {\mathcal{D}}_N}  \EX \left[||f(\Lambdamat )(\tilde{\yvec} - \EX[\tilde{\yvec}]) - \left(\tilde{\xvec} - \EX\left[\tilde{\xvec}\right] \right)||^2\right].
\end{eqnarray}
The solution to \eqref{derivative_opt_f_only} is given by
\begin{equation} \label{opt_f}
    \hat{f}(\Lambdamat) =
    {\text{diag}}(\dvec_{\tilde{\xvec}\tilde{\yvec}})\Dmat_{\tilde{\yvec}\tilde{\yvec}}^{-1},
\end{equation}
where 
\begin{equation} \label{d_def}
	\dvec_{\tilde{\xvec}\tilde{\yvec}} \define {\text{diag}}(\Cmat_{\tilde{\xvec}\tilde{\yvec}}),~~~
	\Dmat_{\tilde{\yvec}\tilde{\yvec}} \define 
	\text{diag}\left(  \text{diag}(\Cmat_{\tilde{\yvec}\tilde{\yvec}} ) \right).
\end{equation}
We assume that $\Dmat_{\tilde{\yvec}\tilde{\yvec}}$ is non-singular.
The GSP-LMMSE estimator is then given by 
\begin{equation}
\label{opt_LMMSE_GSP}
    \hat{\xvec}^{(\text{GSP-LMMSE})}  = \EX[\xvec] +  \Vmat {\text{diag}}(\dvec_{\tilde{\xvec}\tilde{\yvec}})\Dmat_{\tilde{\yvec}\tilde{\yvec}}^{-1}\Vmat^T(\yvec - \EX[\yvec]).
\end{equation}

The GSP-LMMSE estimator from \eqref{opt_LMMSE_GSP} is based on the {\em{diagonal}} of the covariance matrices of $\tilde{\xvec}$ and $\tilde{\yvec}$ that are defined in \eqref{d_def}, which has advantages from a computational point of view, compared with the LMMSE estimator in \eqref{LMMSE} that uses the full covariance matrices of $\xvec$ and $\yvec$. 
The GSP-LMMSE estimator is in fact the diagonal-LMMSE estimator \cite{dd} in the {\em{graph frequency domain}}.
This estimator minimizes the MSE among the linear estimators of $\tilde{\xvec}$, where the estimation matrix that multiplies $\tilde{\yvec}$ is restricted to be a diagonal matrix.

Similar to the sample-LMMSE estimator in Subsection \ref{Monte_Carlo_LMMSE_estimator}, we can use the sample-mean version of the GSP-LMMSE estimator when the statistics are not completely specified.
The idea is to use $P$ random samples of $\xvec$, $\xvec_1,\ldots,\xvec_p$, that are generated from the model in \eqref{Model} to compute the sample-mean vector from \eqref{sample_mean} and the diagonal sample covariance matrices from \eqref{d_def} by
\begin{eqnarray} \label{sample_d}
	\hat{\dvec}_{\tilde{\xvec}\tilde{\yvec}}=
	 \frac{1}{P}\sum_{p=1}^P
	 ( \tilde\xvec_p- \EX[\tilde\xvec])\circ
	 \left( \Vmat^T\gvec(\Lmat,\xvec_p)- \Vmat^T\hat{\yvec}\right)
\end{eqnarray}
and
\begin{eqnarray} \label{sample_D} 
	[\hat{\Dmat}_{\tilde{\yvec}\tilde{\yvec}}]_{i,i} 	 =
    \frac{1}{P} \sum_{p=1}^P\left([\Vmat^T\gvec(\Lmat,\xvec_p)]_i - [\Vmat^T\bar\yvec]_i\right)^2 + [\Cmat_{\tilde{\wvec}\tilde{\wvec}}]_{i,i}, 
\end{eqnarray}
respectively, where $\hat{\Dmat}_{\tilde{\yvec}\tilde{\yvec}}$ is a diagonal matrix, $\tilde{\wvec}$ is the GFT of $\wvec$, and  $\hat{\yvec}$ is defined in \eqref{sample_mean}.
By substituting \eqref{sample_mean}, \eqref{sample_d}, and \eqref{sample_D} in \eqref{opt_LMMSE_GSP}, we obtain the sample-GSP-LMMSE estimator:
\begin{equation} \label{sample_LMMSE_GSP}
    \hat{\xvec}^{(\text{sGSP-LMMSE})}  = \EX[\xvec] +  \Vmat {\text{diag}}(\hat{\dvec}_{\tilde{\xvec}\tilde{\yvec}})\hat{\Dmat}_{\tilde{\yvec}\tilde{\yvec}}^{-1}\Vmat^T(\yvec - \hat{\yvec}),
\end{equation}
 where we assume that 
 $\hat{\Dmat}_{\tilde{\yvec}\tilde{\yvec}}$ is a non-singular matrix.
This approach is summarized in Algorithm \ref{Alg_opt}.
\vspace{-0.25cm}
\begin{algorithm}[hbt]
\SetAlgoLined
	\textbf{Input:}	
	\begin{itemize}
		\item The function $\gvec(\Lmat,\xvec)$.
		\item The Laplacian matrix $\Lmat$.
		\item The distribution of $\xvec$, $f_\xvec$, and its mean, $\EX[\xvec]$. 
		\item The distribution of $\wvec$.
	\end{itemize}
	{\textbf{Algorithm Steps:}}
	\begin{enumerate}
		\setlength\itemsep{0.2em}
		\item Generate $P$ random samples of $\xvec$, with a pdf $f_\xvec$.
		\item Evaluate the sample vectors: $\hat{\yvec}$, $\hat{\dvec}_{\tilde{\xvec}\tilde{\yvec}}$, and $\hat{\Dmat}_{\tilde{\yvec}\tilde{\yvec}}$ from \eqref{sample_mean},
		 \eqref{sample_d}, and
		 \eqref{sample_D}, respectively.
	\end{enumerate}
	\vspace{1mm}
	{\textbf{Output$\colon$}} sample-GSP-LMMSE estimator $\hat\xvec^{(\text{sGSP-LMMSE})} = \EX[\xvec]+ \Vmat\text{diag}(\hat{\dvec}_{\tilde{\xvec}\tilde{\yvec}} ) \hat{\Dmat}_{\tilde{\yvec}\tilde{\yvec}}^{-1}\Vmat^T
	(\yvec - \hat{\yvec}).$
	\caption{Sample-GSP-LMMSE estimator}
	\label{Alg_opt}
\end{algorithm}

\subsection{Advantages and discussion}
\label{prop_subsection}
The GSP representation affords insight into the frequency contents of {\em{nonlinear}} graph signals, such as smoothness and graph-bandlimitness, which can be incorporated into the GSP estimators. In addition, 
a main advantage of the proposed sample-GSP-LMMSE estimator is that it requires the estimation of only two $N$-length vectors that contain the diagonal of the covariance and of the cross-covariance matrices, as described in \eqref{sample_d} and \eqref{sample_D}.
This is in contrast with the sample-LMMSE estimator from \eqref{LMMSE} that requires estimating two $N\times N$ matrices, $\hat{\Cmat}_{\xvec\yvec}$ and $\hat{\Cmat}_{\yvec\yvec}$ from \eqref{hat_C_theta_x} and \eqref{hat_C_y}, respectively.
This advantage improves the sample-GSP-LMMSE estimator performance compared to the sample-LMMSE estimator with  limited datasets used for the non-parametric estimation of the different sample-mean values.
Moreover, the estimation of the inverse sample diagonal covariance matrix, $\hat{\Dmat}_{\tilde{\yvec}\tilde{\yvec}}$ from \eqref{sample_D}, is more robust to limited data (i.e., smaller $P$) than the estimation of the inverse sample covariance matrix, $\hat{\Cmat}_{\yvec\yvec}$ from \eqref{hat_C_y}.
Since $\hat{\Dmat}_{\tilde{\yvec}\tilde{\yvec}}$ is a diagonal matrix, it is a non-singular matrix with probability 1 for any $P\geq 2$ (under the assumption that $\xvec$ is a continuous random variable), while the non-diagonal matrix, $\hat{\Cmat}_{\yvec\yvec}$, requires $P$ to be much larger to be a non-singular matrix.
As a result, in settings where the sample size ($P$) is comparable to the observation dimension ($N$), the sample-LMMSE estimator exhibits severe performance degradation  \cite{4914845,5484583}. This is because the sample covariance matrix is not well-conditioned in the small sample size regime, and inverting it amplifies the estimation error \cite{ledoit2004well}.
 In  the
asymptotic regime where
the observation dimension, $N$, is fixed and $P\rightarrow \infty$, the sample mean and sample covariance matrices  are consistent estimators. Thus, in the asymptotic case, the sample-LMMSE and the sample-GSP-LMMSE estimators converge to the LMMSE and the GSP-LMMSE estimators.

In terms of computational complexity, 
the  sample-LMMSE estimator from \eqref{Monte_Carlo_LMMSE} requires: 
1) forming the sample mean 
and  the sample covariance matrices $	\hat{\Cmat}_{\xvec\yvec}$ and $\hat{\Cmat}_{\yvec\yvec}$ from
\eqref{sample_mean},
\eqref{hat_C_theta_x},
and \eqref{hat_C_y}, respectively,
with full matrix multiplications and an additional cost of $\mathcal{O}(P N^2)$; 2)
computing the inverse of the $N\times N$ matrix $\hat{\Cmat}_{\yvec\yvec}$, which has a complexity of $\mathcal{O}(N^3)$; and  3) performing full matrix multiplications, with a  computational complexity of $\mathcal{O}(N^3)$.
The sample-GSP-LMMSE estimator from \eqref{sample_LMMSE_GSP} requires: 
1) forming the sample mean in
\eqref{sample_mean}
and  the {\em{diagonal}} sample covariance matrices $
	\dvec_{\tilde{\xvec}\tilde{\yvec}} $ and
	$\Dmat_{\tilde{\yvec}\tilde{\yvec}}$ from
\eqref{sample_d} and \eqref{sample_D}, respectively,
with a cost of $\mathcal{O}(P N)$, where the data is generated in the graph frequency domain;
2)
computing the inverse of the diagonal matrix $\hat{\Dmat}_{\tilde{\yvec}\tilde{\yvec}}$, which has a complexity of $\mathcal{O}(N)$;  and 3) performing multiplications of diagonal matrices  with a cost of $\mathcal{O}(N^2)$. The estimator in the vertex domain in \eqref{sample_LMMSE_GSP} requires matrix multiplication by $\Vmat$ and $\Vmat^T$ with a cost of $\mathcal{O}(N^3)$.

The use of the graph frequency domain in the GSP-LMMSE estimator requires the computation of the EVD of the Laplacian matrix, which is of order $\mathcal{O}(N^3)$, and can be computed offline.
If the EVD can be assumed to be known then this task may be avoided.
Recent works propose low-complexity  methods to reduce the complexity of this task (see, e.g. \cite{SVD}). In addition,
the computational complexity of the GSP-LMMSE estimator can be reduced even further 
by using the properties of 
the Laplacian matrix, which tends to be sparse, and therefore,
matrix operations may require fewer computations.

\subsection{Linear optimality conditions} \label{separately_model_section}
In this subsection, we address the question of under which situations the proposed GSP-LMMSE estimator coincides with the LMMSE estimator and with the graphical Wiener filter. 

The following theorem states sufficient and necessary conditions for the GSP-LMMSE and LMMSE estimators to coincide.
\begin{Theorem} \label{claim_coincides}
The GSP-LMMSE estimator coincides with the LMMSE estimator if
\begin{equation} \label{coincides_condition}
 \Cmat_{\tilde{\xvec}\tilde{\yvec}}\Cmat_{\tilde{\yvec}\tilde{\yvec}}^{-1}= 
 	{\text{\normalfont{diag}}}(\dvec_{\tilde{\xvec}\tilde{\yvec}})\Dmat_{\tilde{\yvec}\tilde{\yvec}}^{-1},
\end{equation}
where $\tilde\yvec$ and $\tilde\xvec$ are the GFT of $\yvec$ and $\xvec$, and where $\Cmat_{\tilde{\yvec}\tilde{\yvec}}$ and ${\Dmat}_{\tilde{\yvec}\tilde{\yvec}}$ are non-singular matrices. 
\end{Theorem}
\begin{proof}
By comparing the r.h.s. of \eqref{LMMSE} and the r.h.s. of \eqref{opt_LMMSE_GSP}, it can be verified that the GSP-LMMSE estimator coincides with the LMMSE estimator if
\begin{equation} \label{coincides_Appendix_to_prove_2}
 \Cmat_{\xvec\yvec}\Cmat_{\yvec\yvec}^{-1}= 
    \Vmat 
    {\text{diag}}(\dvec_{\tilde{\xvec}\tilde{\yvec}})\Dmat_{\tilde{\yvec}\tilde{\yvec}}^{-1}\Vmat^T.
\end{equation}
By right and left multiplication of \eqref{coincides_Appendix_to_prove_2} by $\Vmat^T$ and $\Vmat$, respectively, and using $\Vmat^T\Vmat=\Vmat\Vmat^T=\Imat_N$, we obtain
\begin{equation} \label{coincides_Appendix_to_prove_3}
    \Vmat^T \Cmat_{\xvec\yvec}\Vmat\Vmat^T\Cmat_{\yvec\yvec}^{-1}\Vmat= 
  {\text{diag}}(\dvec_{\tilde{\xvec}\tilde{\yvec}})\Dmat_{\tilde{\yvec}\tilde{\yvec}}^{-1}.
\end{equation}
Using the GFT definition in \eqref{GFT}, the condition in 
\eqref{coincides_Appendix_to_prove_3} can be rewritten as \eqref{coincides_condition}.
\end{proof}
The intuition behind the result in Theorem \ref{claim_coincides} is that if \eqref{coincides_condition} is satisfied, then the LMMSE estimator is also the diagonal LMMSE estimator in the graph frequency domain, i.e., the GSP-LMMSE estimator.
By using the definitions in \eqref{d_def}, it can be verified that a sufficient (but not necessary) condition for \eqref{coincides_condition} to hold is that $\Cmat_{\tilde{\xvec}\tilde{\yvec}}$ and $\Cmat_{\tilde{\yvec}\tilde{\yvec}}$ are diagonal matrices. In this case,  $\Cmat_{\tilde{\xvec}\tilde{\yvec}}={\text{diag}}(	\dvec_{\tilde{\xvec}\tilde{\yvec}} )$ and $\Cmat_{\tilde{\yvec}\tilde{\yvec}}=\Dmat_{\tilde{\yvec}\tilde{\yvec}}$.
In the following Theorems we present two special cases for which $\Cmat_{\tilde{\xvec}\tilde{\yvec}}$ and $\Cmat_{\tilde{\yvec}\tilde{\yvec}}$ are diagonal matrices, and thus, \eqref{coincides_condition} holds.
\begin{Theorem}\label{claim_separately_Model}
The GSP-LMMSE estimator coincides with the LMMSE estimator if the following conditions hold:
	\renewcommand{\theenumi}{C.\arabic{enumi}}
\begin{enumerate}
    \item\label{cond1} The nonlinear measurements function, $\gvec(\Lmat,\xvec)$, is separable in the graph frequency domain (``orthogonal frequencies").
    That is, it satisfies
    \begin{equation} \label{separately_Model}
        [\tilde{\gvec}(\Lmat,\xvec)]_n  = [\tilde{\gvec}(\Lmat,\tilde{x}_n\vvec_n)]_n,~n=1,\ldots,N,
    \end{equation}
    where $\tilde{x}_n$ is the $n$th element of $\tilde{\xvec}$ and $\tilde{\gvec}=\Vmat^T\gvec(\Lmat,\xvec)$.
        \item\label{cond2} The elements of the input graph signal, $\xvec$, are statistically independent in the graph frequency domain.
    \item\label{cond3} The noise vector, $\wvec$, is uncorrelated in the graph frequency domain, i.e., $\Cmat_{\tilde{\wvec}\tilde{\wvec}}$ is a diagonal matrix.
\end{enumerate}
\end{Theorem}
\begin{proof}
The proof is given in Appendix \ref{separately_Model_Appendix}.
\end{proof}

\begin{Theorem}\label{claim_graphical_Model}
The GSP-LMMSE estimator coincides with the LMMSE estimator if Condition \ref{cond3} holds and in addition 
	\renewcommand{\theenumi}{C.\arabic{enumi}}
\begin{enumerate}
	 \setcounter{enumi}{3}
    \item\label{cond4} The measurement function, $\gvec(\Lmat,\xvec)$, is the output of  a linear graph filter as defined in \eqref{laplacian_graph_filter}, i.e.,
    \be
    \label{g_filter}
    \gvec(\Lmat,\xvec) = \Vmat f(\Lambdamat)\Vmat^T \xvec.
    \ee
    \item\label{cond5} The covariance matrix of the input graph signal, $\Cmat_{\xvec\xvec}$, is  diagonalizable by the eigenvector matrix of the Laplacian, $\Vmat$, i.e., $\Cmat_{\tilde{\xvec}\tilde{\xvec}}$ is a diagonal matrix.
\end{enumerate}
\end{Theorem}
\begin{proof}
The proof is given in Appendix \ref{graphical_Model_Appendix}.
\end{proof}
The special case in Theorem \ref{claim_graphical_Model} fits the model behind the development of
the graphical Wiener filter \cite{7891646}: First,
Conditions \ref{cond3} and \ref{cond4} imply the linear model assumed in developing the graphical Wiener filter in Subsection V-A in \cite{7891646}.
Second, 
under Condition \ref{cond5}, the signal $ \xvec - \EX[\xvec]$ is a Graph Wide-Sense Stationary (GWSS) signal (see Definition 3 and Theorem 1 in \cite{7891646}), which is the requirement for the graphical Wiener filter.
Therefore, under the conditions of Theorem \ref{claim_graphical_Model}
the GSP-LMMSE estimator coincides with the graphical Wiener filter (Eq. (13) in \cite{7891646}), which is also the LMMSE estimator in this case.
Therefore, the proposed GSP-LMMSE estimator can be interpreted as the graphical Wiener filter, but without assuming a linear model or GWSS signals.

Conditions \ref{cond2} and \ref{cond5} are  common assumptions in the GSP framework \cite{Dong_Vandergheynst_2016,ramezani2019graph}. However, it should be noted  that even if  the elements of the input graph signal, $\xvec$, are statistically independent, they are  not necessarily statistically independent in the graph frequency domain, as required in Condition \ref{cond2}, nor even uncorrelated.
In addition, Condition \ref{cond1} is satisfied if the function $\gvec(\Lmat,\xvec)$ is diagonalized by the eigenvector matrix $\Vmat$ of  $\Lmat$. Condition \ref{cond4}, in which the measurements are obtained as a synthetic output of a graph linear filter as described in \eqref{g_filter},
is a sufficient condition for Condition \ref{cond1}.  Finally, since Condition \ref{cond1} is less restrictive then Condition \ref{cond4}, Condition \ref{cond2} (independency) is more restrictive than Condition \ref{cond5} (decorrelation).

\section{GSP estimators by parametric graph filters} \label{General_design}
In the previous section, we presented a general approach to finding the optimal GSP-LMMSE estimator. However, the GSP-LMMSE estimator is a function of the specific graph structure with fixed dimensions, and thus: 1) it is not optimal when the topology is changed; 2) it is not  adaptive to changes in the number of vertices, $N$, since when vertices are added or removed, 
the dimension of the Laplacian matrix and of the graph signal are changed  and the GSP-LMMSE in \eqref{opt_LMMSE_GSP} is not a valid estimator; and 3) there is no straightforward way to incorporate other graph frequency constraints. Thus, the optimal graph frequency response needs to be redeveloped for any small change in the topology.

In this section, we formulate the problem of designing a graph filter that is MSE-optimal with a specific parametric representation. Then, we develop three implementations of this design by using the three specific graph filters from Section \ref{background_sec}: a linear pseudo-inverse graph filter (Subsection \ref{linear_pseudo_inverse_filter_subsection}), an ARMA graph filter (Subsection \ref{ARMA_filter_subsection}), and a low-rank ARMA graph filter (Subsection \ref{LC_ARMA_filter_subsection}). 
This universal design of graph filters
by fitting the frequency response over a continuous
range of graph frequencies 
based on the graph frequency response is adaptive to topology changes, e.g. in cases when the number of vertices and/or edges is changed. In addition, the parametric representation provides a straightforward way to integrate GSP properties. Finally, since in practice the desired graph frequency response is often approximated by its sample-mean version based on a training dataset, parametrizations can reduce outlier errors and noise effects.

The choice between the different graph filters used in this section can be done according to the trade-off between approximation accuracy, convergence rate, and computational complexity, as detailed in recent literature \cite{Isufi_Leus2017, sparse_paper}. In particular, the parametrization by the linear pseudo-inverse graph filter requires, in general, a higher filter order than the ARMA graph filter, i.e., $K > Q, R$, in order to obtain a good approximation, which may lead to stability problems in large networks. However, finding the ARMA graph filter parameters is based on a nonconvex optimization, while the parameters of the linear pseudo-inverse graph filter have a closed-form solution. 
Using the ARMA filter is known to improve the approximation accuracy and reduce the number of required filter coefficients compared with those of  the linear  (``finite impulse response") filter \cite{sparse_paper}.
The LR-ARMA GSP estimator uses only part of the EVD of $\Lmat$ and requires a lower filter order compared with the ARMA graph filter. However, its performance may be significantly degraded compared to other filters for signals that are not  low-frequency graph signals.

\subsection{General design of the graph frequency response}
\label{General_design_sub}
The graph frequency response of the sample-GSP-LMMSE estimator from \eqref{opt_f} is defined only at graph frequencies $\lambda_n$, $n=1,\dots,N$. 
In this subsection, our goal is to develop a graph-based linear estimator of the form of \eqref{opt_estimator} that minimizes \eqref{Distance}, but where $f(\Lambdamat)$ is restricted to specific parametrizations as a linear and shift-invariant graph filter.
To this end, the MSE-optimal parameter vector, $\alphavec$, for any graph-filter parametrizations, $h(\Lambdamat;\alphavec)$, is found by solving \eqref{derivative_opt_f_only} after substituting the specific parametrizations $f(\Lambdamat)=h(\Lambdamat;\alphavec)$:
\begin{equation} \label{derivative_opt_alpha}
  \hat{\alphavec} = \argmin_{\alphavecsmall\in \Omega_{\alphavecsmall}} 
  \EX \left[||h(\Lambdamat;\alphavec)(\tilde{\yvec} - \EX[\tilde{\yvec}]) - (\tilde{\xvec} - \EX[\tilde{\xvec}]) ||^2\right],
\end{equation}
where $\Omega_{\alphavecsmall}$ is the relevant parameter space, which is defined by the specific choice of graph filter.
This parametric representation of the graph filters interpolates the graph frequency response of the GSP-LMMSE estimator to any graph frequency.
Therefore, when the topology is changed, e.g. by adding/removing vertices, we can substitute the new eigenvalue matrix, $\Lambdamat$, 
in the graph filter to obtain the  approximation, without generating new training data.

The following theorem states the relation between the frequency response of the GSP-LMMSE estimator from \eqref{opt_f} and the optimal parameter vector of a specific parametrization.
\begin{Theorem} \label{claim_equivalent}
The problem in \eqref{derivative_opt_alpha} is equivalent to the following problem:
\begin{eqnarray} \label{filter_coefficients_opt}
    \hat{\alphavec} = 
    \argmin_{\alphavecsmall \in \Omega_{\alphavecsmall}} ||\Dmat_{\tilde{\yvec}\tilde{\yvec}}^{\frac{1}{2}}( 
    \text{diag}(h (\Lambdamat;\alphavec)) - 
     \text{diag}( \hat{f}(\Lambdamat))
    )||^2,
\end{eqnarray}
where $\hat{f}(\Lambdamat)$ and $\Dmat_{\tilde{\yvec}\tilde{\yvec}}$ are given in \eqref{opt_f} and \eqref{d_def}, respectively.
\end{Theorem}
\begin{proof}
The proof is given in Appendix \ref{WLS_formulation}.
\end{proof}

From Theorem \ref{claim_equivalent}, the minimization of the MSE for a specific linear and shift-invariant graph filter in \eqref{derivative_opt_alpha} is equivalent to the minimization of the WLS distance between the desired graph frequency response of the GSP-LMMSE estimator
and the graph frequency responses of the chosen graph filter.
Using the solution of \eqref{filter_coefficients_opt}, the h-GSP-LMMSE estimator with a general parameterized filter $h(\cdot)$ is given by
\begin{equation} \label{GSP_estimator}
    \hat{\xvec}^{(\text{h-GSP})}  = \EX[\xvec] +  \Vmat h(\Lambdamat;\hat{\alphavec})\Vmat^T(\yvec - \hat{\yvec}).
\end{equation}

Similar to the implementation of the sample-LMMSE and the sample-GSP-LMMSE estimators, using the sample-mean version of $\hat{f}(\Lambdamat)$, obtained by 
substituting \eqref{sample_d} and \eqref{sample_D} in \eqref{opt_f}, we obtain that the sample-mean version of the optimal graph filter parameters from \eqref{filter_coefficients_opt}:
\begin{eqnarray} \label{filter_coefficients_opt_s}
    \hat{\alphavec}^{\text{sample}} = 
    \argmin_{\alphavecsmall \in \Omega_{\alphavecsmall}} ||\hat{\Dmat}_{\tilde{\yvec}\tilde{\yvec}}^{\frac{1}{2}}\left( 
    \text{diag}(h (\Lambdamat;\alphavec)) - 
    \hat{\Dmat}_{\tilde{\yvec}\tilde{\yvec}}^{-1}\hat{\dvec}_{\tilde{\xvec}\tilde{\yvec}}
    \right)||^2.
\end{eqnarray}
Then, by substituting $h(\Lambdamat;\hat{\alphavec})=h(\Lambdamat;\hat{\alphavec}^{\text{sample}})$ in \eqref{GSP_estimator}, the sample-mean version of the h-GSP-LMMSE estimator with a general parameterized filter, $h(\cdot)$, is given by
\begin{equation} \label{sGSP_estimator}
    \hat{\xvec}^{(\text{sh-GSP})}  = \EX[\xvec] +  \Vmat h(\Lambdamat;\hat{\alphavec}^{\text{sample}})\Vmat^T(\yvec - \hat{\yvec}).
\end{equation}
The sample-h-GSP estimator for a general graph filter is summarized in Algorithm \ref{algorithm_general}.
\vspace{-0.25cm}
\begin{algorithm}[hbt]
\SetAlgoLined
	\textbf{Input:}	
	\begin{itemize}
		\item The function $\gvec(\Lmat,\xvec)$.
		\item The Laplacian matrix $\Lmat$.
		\item The distribution of $\xvec$, $f_\xvec$, and its mean, $\EX[\xvec]$. 
		\item The distribution of $\wvec$.
		\item Parameterized  filter $h(\cdot)$.
	\end{itemize}
	{\textbf{Algorithm Steps:}}
	\begin{enumerate}
		\setlength\itemsep{0.2em}
		\item Generate $P$ random samples of $\xvec$, with a pdf $f_\xvec$.
		\item Evaluate the sample vectors: $\hat{\yvec}$, $\hat{\dvec}_{\tilde{\xvec}\tilde{\yvec}}$, and $\hat{\Dmat}_{\tilde{\yvec}\tilde{\yvec}}$ from \eqref{sample_mean},
		 \eqref{sample_d}, and
		 \eqref{sample_D}, respectively. \label{sample_step}
		\item Compute the optimal graph filter coefficient vector, $\hat{\alphavec}^{\text{sample}}$ by solving \eqref{filter_coefficients_opt_s}. \label{Compute_step}
	\end{enumerate}
	\vspace{1mm}
	{\textbf{Output$\colon$}} sample-h-GSP estimator $\hat{\xvec}^{(\text{sh-GSP})}  = \EX[\xvec] +  \Vmat h(\Lambdamat;\hat{\alphavec}^{\text{sample}})\Vmat^T(\yvec - \hat{\yvec}).$
	\caption{Sample-h-GSP estimator}
	\label{algorithm_general}
\end{algorithm}
\vspace{-0.25cm}

In the following subsections, we present three different GSP estimators for different choices of typical graph filters $h(\cdot;\alphavec)$ and we evaluate their associated optimal parameters.
The difference between the estimators in Subsections \ref{linear_pseudo_inverse_filter_subsection}-
\ref{LC_ARMA_filter_subsection} is that
they are based on different parametrizations and implementations. 
The minimization in \eqref{filter_coefficients_opt} implies that
   one should choose a graph-filter parametrization, $h(\Lambdamat;\alphavec)$, that will be close to the desired graph frequency response, $\hat{f}(\Lambdamat)$, from \eqref{opt_f}.
   Therefore, the  shape of the   desired graph frequency response curve affects the suitable choice of the graph filter parameterization,
   which is application-dependent. 
   The study of $\hat{f}(\Lambdamat)$ 
   is expected to lead to a better understanding of what kinds of filters are useful. 
   This study can be done, for example, based on theoretical properties such as low-pass and high-pass graph filters  \cite{Rama2020Anna} and on a simulation study.  
For example, in 
the simulations in Section \ref{simulation}, $\hat{f}(\Lambdamat)$ behaves as  a low-pass graph filter, and thus, we choose graph filters suitable for this case. The presented framework can be easily extended to other graph filters  for different shapes of $\hat{f}(\Lambdamat)$.

\subsection{Filter 1: sample linear pseudo-inverse GSP estimator}
\label{linear_pseudo_inverse_filter_subsection}
For the linear pseudo-inverse graph filter from Subsection \ref{LFI_representation}, the graph filter from \eqref{pseudo_inverse_filter_def} satisfies
\begin{equation} \label{new h}
	\text{diag}(h^{\text{LPI}}(\Lambdamat^{\dagger};\alphavec^{\text{LPI}})) = \bar\Gammamat_K\alphavec^{\text{LPI}},
\end{equation}
 where $\alphavec^{\text{LPI}}$ is the filter coefficients vector from \eqref{LPI_coefficients} and $\bar\Gammamat_K$ is a $N \times (K+1)$ matrix with elements
\begin{equation} \label{barGammamat}
	[\bar\Gammamat_K]_{i,j} = 
		\begin{cases}
\lambda_i^{-(j-1)}, &\text{for } 2 \leq i \leq N,~ j =1,\dots,K+1\\
1, &\text{for } i=j =1\\
0, &\text{otherwise.} 
\end{cases}
\end{equation}
Therefore, the optimal filter coefficients, $\alphavec^{\text{LPI}}$, are obtained by substituting \eqref{new h} in \eqref{filter_coefficients_opt_s} and removing a constant term, which results in
\begin{equation} \label{Linear pseudo-inverse graph filter objective function}
	\hat{\alphavec}^{\text{LPI}}
	= \argmin_{\alphavecsmall\in \mathbb{R}^{K+1}}
	 \alphavec^T\bar{\Gammamat}_K^T\hat{\Dmat}_{\tilde{\yvec}\tilde{\yvec}}\bar{\Gammamat}_K \alphavec
	 - 2\hat{\dvec}_{\tilde{\xvec}\tilde{\yvec}}^T\bar{\Gammamat}_K\alphavec.
\end{equation}

To avoid overfitting, we replace \eqref{Linear pseudo-inverse graph filter objective function} by the following regularized minimization:
\begin{equation} \label{minimization_LPI_final}
	\hat{\alphavec}^{\text{LPI}}
	= \argmin_{\alphavecsmall\in \mathbb{R}^{K+1}}
	 \alphavec^T\bar{\Gammamat}_K^T\hat{\Dmat}_{\tilde{\yvec}\tilde{\yvec}}\bar{\Gammamat}_K \alphavec
	 - 2\hat{\dvec}_{\tilde{\xvec}\tilde{\yvec}}^T\bar{\Gammamat}_K\alphavec  + \mu \alphavec^T\Mmat_{\text{LPI}}\alphavec,
\end{equation}
where $\mu \geq 0$ is a regularization coefficient and $\Mmat_{\text{LPI}}$ is a positive semidefinite regularization matrix. 
By equating the gradient of \eqref{minimization_LPI_final} w.r.t. $\alphavec$ to zero, we obtain
\begin{equation} \label{LPI_invers_Matrix}
     \hat{\alphavec}^{\text{LPI}} = \left(\bar{\Gammamat}_K^T\hat{\Dmat}_{\tilde{\yvec}\tilde{\yvec}}\bar{\Gammamat}_K + \mu\Mmat_{\text{LPI}} \right )^{-1}\bar{\Gammamat}_K^T
\hat{\dvec}_{\tilde{\xvec}\tilde{\yvec}}.
\end{equation}
In practice, in order to avoid the numerically unstable problem of inverting the $(K+1)\times (K+1)$ matrix, one can solve the convex optimization problem in \eqref{minimization_LPI_final} by using any existing quadratic programming algorithm.
The sample-h-GSP estimator with the linear pseudo-inverse graph filter is implemented by Algorithm \ref{algorithm_general}, where Step \ref{Compute_step} is obtained by evaluating the matrix $\bar\Gammamat$ from \eqref{barGammamat} and then computing $\hat{\alphavec}^{\text{LPI}}$ either by \eqref{LPI_invers_Matrix} or by using a quadratic programming algorithm to solve \eqref{minimization_LPI_final}.

Since the regularization matrix, $\Mmat_{\text{LPI}}$, is a positive semidefinite matrix, 
the optimization problem in \eqref{minimization_LPI_final} is a convex optimization problem and $\hat{\alphavec}^{\text{LPI}}$ is its unique solution, for any $\mu \geq 0$ as long as the matrix $\bar{\Gammamat}_K^T\hat{\Dmat}_{\tilde{\yvec}\tilde{\yvec}}\bar{\Gammamat}_K$ is a positive definite matrix. This condition
holds if $\hat{\Dmat}_{\tilde{\yvec}\tilde{\yvec}}$ is a positive definite matrix and $\text{rank}(\bar{\Gammamat}_K) = K+1$ (see e.g. Chapter 7 in \cite{Horn_Johnson_book}).
It is assumed in this paper that $\Dmat_{\tilde{\yvec}\tilde{\yvec}}$ from 
\eqref{d_def} is non-singular (and thus, positive definite). 
Thus, by taking enough off-line measurements, $\xvec_1,\ldots,\xvec_P$,
the sample covariance matrix
 $\hat\Dmat_{\tilde{\yvec}\tilde{\yvec}}$ is a positive definite matrix as well.
The following claim states a condition for 
$\bar{\Gammamat}_K$ from \eqref{barGammamat} to be full (column) rank. 
\begin{Claim}
\label{claim1}
If there are $K+1$ distinct eigenvalues of the matrix $\Lmat$ such that $\lambda_n \neq \lambda_k$, $\forall n \neq k$,
then, $\text{rank}(\bar{\Gammamat}_K) = K+1 $.
\label{diff_lambda_i_2}
\end{Claim}
 \begin{proof}
See Appendix \ref{rank_Appendix}.
\end{proof}

\subsection{Filter 2: sample ARMA GSP estimator} \label{ARMA_filter_subsection}
For the ARMA graph filter from Subsection \ref{ARMA_representation}, the graph filter  \eqref{eq:0} satisfies
\beqna \label{ARMA filter as matrix}
	 \text{diag}(h^{{\text{ARMA}}}(\Lambdamat;\alphavec^{{\text{ARMA}}})) \hspace{3cm}
	 \nonumber\\= \left(\text{diag}(\Phivec(N,R)\avec)\right)^{-1}\Phivec(N,Q)\cvec,
\eeqna
where
$\alphavec^{{\text{ARMA}}}$ is the filter coefficients vector from \eqref{ARMA_coefficients}, 
$\Phivec(N,O)$ is a $N \times (O+1)$ Vandermonde matrix defined by
\begin{equation} \label{Gammamat_def}
	\Phivec(N,O)\define
\left[\begin{array}{cccc}
		1 & \lambda_1 &\dots & \lambda_1^{O} \\
		\vdots & \vdots & \ddots & \vdots\\
		1 & \lambda_N & \dots & \lambda_N^{O}
	\end{array}\right].
\end{equation}
The associated filter coefficients, $\alphavec^{\text{ARMA}} = [\avec^T,\cvec^T]^T$, are obtained by substituting \eqref{ARMA filter as matrix} in \eqref{filter_coefficients_opt_s} and removing a constant term, which results in
\beqna\label{minimization ARMA}
	\left(\hat\avec,\hat\cvec\right) = 
	\argmin_{\substack{\avec\in \mathbb{R}^{R+1}\\\cvec\in \mathbb{R}^{Q+1}}} 
\left\{	\cvec^T\Phivec^T(N,Q)\left(\text{diag}(\Phivec(N,R)\avec)\right)^{-1}\right.\nonumber
\\
\times
	\hat{\Dmat}_{\tilde{\yvec}\tilde{\yvec}} \left(\text{diag}(\Phivec(N,R)\avec)\right)^{-1} \Phivec(N,Q)\cvec \nonumber
	\\
	- 2\hat{\dvec}_{\tilde{\xvec}\tilde{\yvec}}^T  \left(\text{diag}(\Phivec(N,R)\avec)\right)^{-1}\Phivec(N,Q)\cvec  \nonumber\\\left.
	 + \mu\avec^T\Mmat_{\avec}\avec  
	 +\mu\cvec^T\Mmat_{\cvec}\cvec\right\}, 
\eeqna
where $a_0 =1$ and the last two terms are regularization terms that have been added to avoid overfitting, in which
$\mu \geq 0$ is a regularization coefficient, and $\Mmat_{\avec}$ and $\Mmat_{\cvec}$ are positive semidefinite regularization matrices. 
Equating the derivative of \eqref{minimization ARMA} w.r.t. $\cvec$ to zero, results in
\begin{equation} \label{b ARMA}
\begin{split}
\hat{\cvec}^{\text{ARMA}}(\avec) = \hspace{6cm}\\
( \Phivec^T(N,Q)(\text{diag}(\Phivec(N,R)\avec))^{-1}
\hat{\Dmat}_{\tilde{\yvec}\tilde{\yvec}}(\text{diag}(\Phivec(N,R)\avec))^{-1} \\ 
\times \Phivec(N,Q) +\mu\Mmat_{\cvec})^{-1}
\Phivec^T(N,Q)(\text{diag}(\Phivec(N,R)\avec))^{-1}\hat{\dvec}_{\tilde{\xvec}\tilde{\yvec}}.
\end{split}
\end{equation}
By substituting \eqref{b ARMA} in the objective function from \eqref{minimization ARMA}, we obtain 
\begin{eqnarray}\label{minimization_ARMA_new}
	\hat\avec^{\text{ARMA}} 	= \argmax_{\avec\in \mathbb{R}^{R+1}} 
	\hat{\dvec}_{\tilde{\xvec}\tilde{\yvec}}^T(\text{diag}(\Phivec(N,R)\avec))^{-1}\hspace{1cm}
	\nonumber\\
	\times \Phivec(N,Q)
	\hat{\cvec}^{\text{ARMA}}(\avec)
	+ \mu\avec^T\Mmat_{\avec}\avec, 
\end{eqnarray}
where $a_0 =1$.
The optimal $\cvec$ is given by substituting the solution of \eqref{minimization_ARMA_new} in \eqref{b ARMA}, i.e., $\hat{\cvec}^{\text{ARMA}} = \cvec(\hat{\avec}^{\text{ARMA}})$.
Finally, the sample-GSP estimator with the ARMA graph filter is implemented by Algorithm \ref{algorithm_general}, where Step \ref{Compute_step} is obtained by: I. evaluating the matrices $\Phivec(N,R)$ and $\Phivec(N,Q)$ from \eqref{Gammamat_def}; II. computing $\hat{\avec}$ by solving \eqref{minimization_ARMA_new} numerically; and III. Computing $\hat{\cvec}$ by substituting the result of II in \eqref{b ARMA}. In the simulations we used the Matlab function `fminsearch' to approximate \eqref{minimization_ARMA_new}.

Since the regularization matrix, $\Mmat_{\cvec}$, is a positive semidefinite matrix,  \eqref{minimization ARMA} is a convex optimization problem w.r.t. $\cvec$ for any $\mu\geq 0$ as long as  
$\hat{\Dmat}_{\tilde{\yvec}\tilde{\yvec}}$ is  positive definite  and $\text{rank}((\text{diag}(\Phivec(N,R)\avec))^{-1} \Phivec(N,Q)) = Q+1$.
The following claim 
states the condition for this matrix to be  full rank.
\begin{Claim}
\label{claim2}
If there are $Q+1$ distinct eigenvalues of the matrix $\Lmat$ and $[\Phivec(N,R)\avec]_n \neq 0$, $\forall n=1,\dots,N$, then, $\text{rank}((\text{diag}(\Phivec(N,R)\avec))^{-1} \Phivec(N,Q)) = Q+1 $.
\end{Claim}
 \begin{proof}
See Appendix \ref{rank_Appendix}.
\end{proof}
Therefore, if $\hat{\Dmat}_{\tilde{\yvec}\tilde{\yvec}}$ is a positive definite matrix and the condition in Claim \ref{claim2} holds,
then the objective function in \eqref{minimization ARMA} is a convex optimization problem w.r.t. $\cvec$ and $\hat{\cvec}^{\text{ARMA}}(\avec)$ is its unique solution. 
It should be noted that the ARMA graph filter involves Vandermonde matrices; 
in order to obtain a stable solution, $R$ and $Q$ should be chosen to have small values \cite{sparse_paper}.

As explained after \eqref{Linear_filter_def},
   the linear graph filter is
 a special case of the ARMA graph filter.
 Thus, 
the optimal coefficients  of the linear graph filter  are obtained by substituting $\avec = a_0 = 1$, $R=0$,   $Q=K$, and $\mu=0$ in 
\eqref{b ARMA}, which results in
\begin{eqnarray}\label{Linear_coefficients_opt}
     \hat{\alphavec}^{\text{lin}} =\hat{\cvec}^{\text{ARMA}}(\avec= 1) \hspace{4.5cm}\nonumber\\=
    \left(\Phivec(N,K)^T\hat{\Dmat}_{\tilde{\yvec}\tilde{\yvec}}\Phivec(N,K) \right )^{-1}\Phivec(N,K)^T
\hat{\dvec}_{\tilde{\xvec}\tilde{\yvec}},
\end{eqnarray}
where $\Phivec(\cdot,\cdot)$ is the  Vandermonde matrix  defined in \eqref{Gammamat_def}.
While all  filter designs  have an equivalent polynomial filter, the matrix $\Phivec(N,K)$ needs to be well-conditioned in order to obtain good estimation by using the filter in \eqref{Linear_filter_def} with the coefficients in  \eqref{Linear_coefficients_opt}. This will only be the case for small graph sizes $N$ and/or small filter orders $K$ \cite{sparse_paper,6808520,7926424}, which leads to limited accuracy of the linear graph filter.

\subsection{Filter 3: sample low-rank ARMA GSP estimator} \label{LC_ARMA_filter_subsection}
For the LR-ARMA graph filter from Subsection \ref{LR_representation}, the graph filter from \eqref{LR_ARMA_def} satisfies
\begin{eqnarray}\label{LC_ARMA_filter_as_matrix}
	 h^{{\text{LR-ARMA}}}(\bar{\Lambdamat};\alphavec^{\text{LR}}) \hspace{5cm}\nonumber\\
	 =\left[
	    ( (\text{diag}(\Phivec(N_s,R)\avec^{{\text{LR}}}))^{-1}\Phivec(N_s,Q)\cvec^{{\text{LR}}})^T, \zerovec_{N-N_s}^T
	\right]^T
	 ,
\end{eqnarray}
where
$\alphavec^{{\text{LR}}}= [\avec^{{\text{LR}}^T},\cvec^{{\text{LR}}^T}]^T$ is the filter coefficients vector from \eqref{LR_ARMA_coefficients}, and 
$\Phivec(\cdot,\cdot)$ is the  Vandermonde matrix  defined in \eqref{Gammamat_def}. 
Let $\mathcal{U} = \{1,\dots,N_s\}$, $\hat{\Dmat}_{{\tilde{\yvec}\tilde{\yvec}}_{\mathcal{U}}}$ denotes the matrix that includes the first $N_s$ rows and columns of $\hat{\Dmat}_{{\tilde{\yvec}\tilde{\yvec}}}$ and $\hat{\dvec}_{{\tilde{\xvec}\tilde{\yvec}}_\mathcal{U}}$ denotes the vector that includes the first $N_s$ elements of $\hat{\dvec}_{{\tilde{\xvec}\tilde{\yvec}}}$.
Then, 
similarly to \eqref{b ARMA}-\eqref{minimization_ARMA_new},
the coefficient vector, $\alphavec^{\text{LR}} $, is obtained by  minimizing \eqref{filter_coefficients_opt_s} with a regularization term after the substitution of \eqref{LC_ARMA_filter_as_matrix}, where $\Dmat_{\tilde{\yvec}\tilde{\yvec}}$ is a diagonal matrix and the last $N-N_s$ entries of \eqref{LC_ARMA_filter_as_matrix} are zero,  
 which results in
\begin{eqnarray} \label{b LR-ARMA}
\hat{\cvec}^{\text{LR}}(\avec) = \left( \Phivec^T(N_s,Q)(\text{diag}(\Phivec(N_s,R)\avec))^{-1}\hat{\Dmat}_{{\tilde{\yvec}\tilde{\yvec}}_{\mathcal{U}}}
\nonumber\right.\\\left.  \times
(\text{diag}(\Phivec(N_s,R)\avec))^{-1}\Phivec(N_s,Q) +\mu\Mmat_{\cvec^{\text{LR}}}
\vphantom{
(\Phivec(N_s,Q))^T(\text{diag}(\Phivec(N_s,R)\avec))^{-1}\Dmat_{\tilde{\yvec}\tilde{\yvec}}(\text{diag}(\Phivec(N_s,R)\avec))^{-1}\Phivec(N_s,Q) 
+\mu\Mmat_{\cvec^{\text{LR}}}
}
\right)^{-1}
\nonumber \hspace{-0.98cm}\\ \times
\Phivec^T(N_s,Q)\left(\text{diag}(\Phivec(N_s,R)\avec)\right)^{-1}
\hat{\dvec}_{{\tilde{\xvec}\tilde{\yvec}}_\mathcal{U}}, 
\end{eqnarray}
and
\begin{eqnarray} \label{minimization LR-ARMA new}
	\hat\avec^{\text{LR}} 	= \argmin_{\avec\in \mathbb{R}^{R+1}} 
	\hat{\dvec}_{{\tilde{\xvec}\tilde{\yvec}}_\mathcal{U}}^T\left(\text{diag}(\Phivec(N_s,R)\avec)\right)^{-1}\Phivec(N_s,Q) \hat{\cvec}^{\text{LR}}(\avec)\hspace{-0.98cm} \nonumber\\
	+ \mu\avec^T\Mmat_{\avec^{\text{LR}}}\avec, \hspace{3cm}
\end{eqnarray}
where $a_0 =1$
and $\mu \geq 0, \Mmat_{\avec^{\text{LR}}}, \Mmat_{\cvec^{\text{LR}}}  0$ are regularization coefficient and positive semidefinite matrices. 
Then, the optimal $\cvec$ is given by substituting the solution of \eqref{minimization LR-ARMA new} in \eqref{b LR-ARMA}, i.e., $\hat{\cvec}^{\text{LR}} = \cvec(\hat{\avec}^{\text{LR}})$.
Similar to Claim \ref{claim2}, the conditions for convexity w.r.t. $\cvec$ can be derived.
Finally, the sample-GSP estimator with the LR-ARMA graph filter is implemented by Algorithm \ref{algorithm_general}, wherein Step \ref{sample_step} evaluates the sample subvectors, $\hat{\yvec}_\mathcal{U}$, $\hat{\dvec}_{{\tilde{\xvec}\tilde{\yvec}}_\mathcal{U}}$, and $\hat{\Dmat}_{{\tilde{\yvec}\tilde{\yvec}}_{\mathcal{U}}}$ from \eqref{sample_mean}, \eqref{sample_d}, and \eqref{sample_D}, respectively; and Step \ref{Compute_step} is obtained by evaluating the matrices $\Phivec(N_s,R)$ and $\Phivec(N_s,Q)$ from \eqref{Gammamat_def} and computing $\hat{\cvec}^{\text{LR}}(\avec)$ and  $\hat{\avec}^{\text{LR}}$ by solving \eqref{b LR-ARMA} and \eqref{minimization LR-ARMA new}, respectively.

\subsection{Advantages and discussion} \label{computational complexity}
An important advantage of the sample linear pseudo-inverse GSP estimator is that its parameters have a closed-form  analytic expression in \eqref{LPI_invers_Matrix}. In contrast, 
the evaluation of the sample ARMA and the LR-ARMA GSP estimators  requires solving nonconvex optimization problems \eqref{minimization_ARMA_new} and \eqref{minimization LR-ARMA new}, respectively. 
On the other hand, since the graph frequency response of the 
linear pseudo-inverse graph filter from \eqref{LPI graph filter} has a  discontinuity  at $\lambda = 0$, it may be unstable when   $\lambda_2$ approaches $0$ due to topology changes.
Since $\lambda_2$ describes the graph connectivity \cite{Newman_2010}, this is only a problem when the network becomes disconnected. 
The relation between the spectrum
of the Laplacian matrix of graphs and the graph  is extensively discussed in the literature (see, e.g. in \cite{oellermann1991laplacian2}). For example, suppose the desired graph frequency response, $\hat{f}(\Lambdamat)$, from \eqref{opt_f}  and the 
graph-filter parametrization, $h(\Lambdamat;\alphavec)$, 
satisfy some smoothness assumptions. In this case, it can be shown by using results from \cite{Isufi_Ribeiro,9206091,gao2021stability,kenlay2021interpretable} that the difference between the optimal graph frequency response {\em{after the change}} and the  graph-filter parametrization, $h(\Lambdamat;\alphavec)$, that is computed with the new $\Lambdamat$, is bounded when the topology change is bounded.

In terms of computational complexity, 
the sample-h-GSP from \eqref{sGSP_estimator}  computes the same expressions as the sample-GSP-LMMSE estimator from \eqref{sample_LMMSE_GSP},  and thus, has the same computational complexity as described in Subsection \ref{prop_subsection},  with an additional complexity  that stems from:
a) performing the matrix multiplications of $\Vmat h(\Lambdamat;\hat{\alphavec}^{\text{sample}})\Vmat^T$, with a  computational complexity of $\mathcal{O}(N^3)$;
and b) the implementation of the specific graph filter.
In detail, 
 implementing the sample linear pseudo-inverse GSP estimator from Subsection \ref{linear_pseudo_inverse_filter_subsection} requires:
1) evaluating the optimal filter coefficients $\hat{\alphavec}^{\text{LPI}}$ from \eqref{LPI_invers_Matrix} by computing the inverse of a $K+1 \times K+1$ matrix, with complexity $\mathcal{O}(K^3)$, where $K \ll N$; and
2) computing the graph frequency response from \eqref{new h} with a matrix-vector multiplication
of $\mathcal{O}(NK)$.
Implementing the sample ARMA and LR-ARMA GSP estimators from Subsections \ref{ARMA_filter_subsection} and \ref{LC_ARMA_filter_subsection}, respectively, requires:
1) evaluating the optimal filter coefficients by solving the nonconvex optimization problem from \eqref{minimization_ARMA_new} and \eqref{minimization LR-ARMA new}, respectively, which has a complexity that depends on the chosen optimization algorithm; and
2) computing the graph frequency response from \eqref{ARMA filter as matrix} and \eqref{LC_ARMA_filter_as_matrix}, respectively, with a matrix-vector multiplication and computing the inverse of the diagonal matrix, with a cost of $\mathcal{O}(N(K+1))$.
Finally, all GSP estimators 
require  the EVD of the Laplacian matrix for computing $\Vmat$. This typically requires a $\mathcal{O}(N^3)$ complexity cost, but several fast computation methods  for spectral
decompositions \cite{le2017approximate,lu2019fast} can be used.

The main advantage of the sample-h-GSP estimators is their low computational complexity needed for updating the estimators when the topology changes. In this case, the updated estimator is evaluated by using the graph filter coefficients, $\hat{\alphavec}^{\text{sample}}$, that were evaluated based on the original topology, with the Laplacian matrix of the new topology.
As a result, the updated sample-h-GSP estimator from \eqref{sGSP_estimator} only requires reevaluation of $h(\Lambdamat;\hat{\alphavec}^{\text{sample}})$, where the graph filter coefficients, $\hat{\alphavec}^{\text{sample}}$, are known, which has a maximum complexity of $\mathcal{O}(N^3)$. 
This is in contrast with the reevaluation needed for the sample-LMMSE and the sample-GSP-LMMSE estimators. 
In addition, the low complexity and distributed implementation of the sample ARMA GSP estimator is described in \cite{sparse_paper,7001054}.

\section{Simulation} \label{simulation}
In this section we evaluate the performance of the proposed GSP-LMMSE estimator from Section \ref{GSP_estimator_section} and the three parametrizations from Subsections \ref{linear_pseudo_inverse_filter_subsection}-\ref{LC_ARMA_filter_subsection} for solving the problem of power system state estimation (PSSE), which is essential for various monitoring purposes 
\cite{Abor}.
The setting of the PSSE problem is presented in Subsection \ref{seeting_sec}.
The different estimation methods that are presented in this section are described in Subsection \ref{Methods}.
The results for stationary networks and for networks with topology changes are presented in Subsections \ref{subsection_statonary}  and \ref{subsection_change}, respectively.

\subsection{Case study: PSSE in electrical networks}
\label{seeting_sec}
A power system can be represented as an undirected weighted graph, ${\pazocal{G}}({\pazocal{V}},\xi)$, where the set of vertices, $\pazocal{V}$, is the set of buses (generators or loads) and the edge set, $\xi$, is the set of transmission lines between these buses.
The measurement vector of the active powers at the buses, $\yvec$,
can be described by the model in \eqref{Model}, with nonlinear measurement function
\begin{eqnarray} \label{g_AC}
\left[\gvec(\Lmat,\xvec)\right]_n\
 \define \sum_{m=1}^N |v_n||v_m|(G_{n,m}\cos(x_n -x_{m})
\nonumber\\+B_{n,m}\sin(x_n -x_{m})),\hspace{1.5cm}
\end{eqnarray}
$n=1,\ldots,N$. Here $x_n$ and $|v_n|$ are the voltage phase and amplitude at the $n$th bus,
and $G_{n,m}$ and $B_{n,m}$ are the conductance and susceptance of the transmission line between the  buses $n$ and $m$ \cite{Abor}, where $(n,m)\in\xi$.
In the graph modeling of the electrical network,
the Laplacian matrix, $\Lmat$, is constructed by using $B_{n,m}$, $n,m=1,\ldots,N$ (see Subsection II-C in \cite{Grotas_2019}).
We assume that $|v_n|=1$, which is a common assumption \cite{Abor}, and $G_{n,m}$ and $B_{n,m}$ are all known.

The goal of PSSE is to recover the state vector, $\xvec$, from the power measurements of $\gvec(\Lmat,\xvec)$, which is known to be a NP-hard problem \cite{bienstock2019strong}.
The input graph signal, $\xvec$, is shown to be smooth \cite{drayer2018detection,dabush2021state}, i.e., its graph smoothness \cite{Shuman_Ortega_2013,8347162}, is small.
Therefore, we model the distribution of the input graph signal, $\xvec$, in the graph frequency domain \cite{Dong_Vandergheynst_2016,ramezani2019graph}, as a smooth Gaussian distribution, as follows:
\begin{equation} \label{theta_distribution_simulation}
		\tilde{\xvec}_{\text{2$\colon$end}} \sim \pazocal{N}(\zerovec,\beta \Lambdamat^{-1}_{\text{2$\colon$end},\text{2$\colon$end}}),
\end{equation}
where $\beta$ is a  smoothness level.
The smooth distribution of $\tilde{\xvec}$ from \eqref{theta_distribution_simulation} implies in particular that the first graph frequency of $\xvec$ satisfies $\tilde{x}_1 =0$.
Finally,
we assume that the noise term, $\wvec$, from the model in \eqref{Model} in this case is zero-mean Gaussian with covariance matrix $\Cmat_{\wvec\wvec}=\sigma^2\Imat_N$.
It can be seen that for this example, Conditions \ref{cond2}, \ref{cond3}, and \ref{cond5} from Theorems \ref{claim_separately_Model} and \ref{claim_graphical_Model} are satisfied, but Conditions \ref{cond1} and \ref{cond4} are not satisfied. It can also be shown that \eqref{coincides_Appendix_to_prove_2} is not satisfied. 
Therefore,  the proposed approach is neither the LMMSE estimator nor the graphical Wiener filter \cite{7891646} in this case.
The values of the different physical parameters in \eqref{g_AC} are taken from the test case of a 118-bus IEEE power system \cite{iEEEdata}, where $N=118$.
The MSE of the different estimators is calculated by performing 10,000 Monte Carlo simulations.

\subsection{Methods} \label{Methods}
In the simulations we compare the performance of the following estimators:\\
1) The sample-LMMSE estimator from  \eqref{Monte_Carlo_LMMSE}.\\
    2) The sample-LMMSE estimator from  \eqref{Monte_Carlo_LMMSE} with a large  $P$ ($P=5 \cdot 10^6 \hspace{-0.15cm}\gg \hspace{-0.15cm}N$),  denoted as $P_\infty$-LMMSE. 
 In  this
asymptotic regime, the sample-LMMSE  estimator converges to the LMMSE  estimator. Since 
the MSE of the LMMSE estimator is  lower than   the MSE of any linear estimator, it can be used as a benchmark 
for a stationary network.\\
3)  The sample-GSP-LMMSE estimator from Algorithm \ref{Alg_opt}.\\
4)  The sample linear pseudo-inverse GSP estimator from Algorithm \ref{algorithm_general} with $K=6$, where Step \ref{Compute_step} is implemented as explained after \eqref{LPI_invers_Matrix}, and the regularization matrix is set to $
        \Mmat_{\text{LPI}} = \text{diag}(\lambda_N^0,\dots,\lambda_N^{K})$ in order to restrict the length of the filter and of the power of the pseudo-inverse of the Laplacian.\\
5) The sample ARMA GSP estimator, 
which is implemented as explained after \eqref{minimization_ARMA_new}, with $R,Q = 3$, $\Mmat_{\avec} = \Imat_{R}$, $\Mmat_{\cvec} = \Imat_Q$. \\
6) The sample LR-ARMA GSP estimator, 
which is implemented as explained after \eqref{minimization LR-ARMA new}, with ${R}_{{\text{LR}}},Q_{{\text{LR}}} = 2$, $\Mmat_{\avec^{\text{LR}}}= \Imat_{{R}_{{\text{LR}}}}$, $\Mmat_{\cvec^{\text{LR}}}= \Imat_{Q_{{\text{LR}}}}$, and the cutoff frequency $N_s = 0.3N$.
    Thus, it  uses only the $30\%$ smallest eigenvalues and their associated eigenvectors of the Laplacian matrix, $\Lmat$.\\
7) The LMMSE estimator evaluated for the linear approximation of the model in \eqref{g_AC}. That is, the nonlinear function from \eqref{g_AC} is linearized by \cite{Abor} 
      $ \gvec(\Lmat,\xvec) \approx \Lmat\xvec$. 
    Then, 
    $\EX[\yvec] = \zerovec$,
    $\Cmat_{\yvec\yvec} = \beta\Lmat + \sigma^2\Imat_N $,
    and
    $ \Cmat_{\xvec\yvec} = \beta\Lmat^{\dagger}\Lmat$,
    resulting in \begin{equation} \label{LS}
	    \hat\xvec^{(\text{aLMMSE})} =\beta \Lmat^\dagger\Lmat(\beta\Lmat + \sigma^2\Imat_N)^{-1}\yvec.
    \end{equation}
    
\subsection{Example A: stationary network}
\label{subsection_statonary}
In this subsection, we investigate the case where the topology is constant and, thus, the statistical properties of $\yvec,\xvec$ are constant.
In Fig. \ref{fig:Diff N}, we present the MSE of the methods from Subsection \ref{Methods} for different values of $P$, i.e., different numbers of training data points used to evaluate the sample-mean values.
It can be seen that the linearized-model based estimator from \eqref{LS} and the lower bound obtained by  $P_\infty$-LMMSE are independent of $P$, as expected.
The sample-LMMSE estimator from \eqref{Monte_Carlo_LMMSE} uses the inverse of the sample covariance matrix $\hat{\Cmat}_{\yvec\yvec}$, which requires a large number of training data points to achieve a stable estimation. Thus,
it can be seen that for $P<10N$, where $N=118$, the MSE of the sample-LMMSE estimator is higher than the MSE of the proposed methods: the sample-GSP-LMMSE, the sample linear pseudo-inverse GSP, and the sample ARMA GSP estimators.

In this example, the sample linear pseudo-inverse GSP  and the sample ARMA GSP estimators coincide with the  sample-GSP-LMMSE estimator.
Thus, the chosen parametrizations are an accurate approximation of the desired graph frequency response. In addition,
for $P>10^4$, these GSP estimators and the  sample-LMMSE estimator converge. It can be seen that the MSE of the LMMSE estimator, represented by $P_\infty$-LMMSE, provides a lower bound on the MSE of any linear estimator, where  the GSP-LMMSE, linear pseudo-inverse GSP, and ARMA GSP estimators achieve this lower bound for a much smaller value of $P$ than the sample-LMMSE estimator. This result holds although Condition \ref{cond1}  from Theorem \ref{claim_separately_Model} is not satisfied and the proposed GSP-LMMSE estimator differs from the LMMSE estimator.
Finally, the sample LR-ARMA GSP estimator has a lower MSE than the sample-LMMSE estimator for $P<3N$ and achieves lower MSE than the LMMSE estimator evaluated for the linear approximation from \eqref{LS} for $P>0.1N$. 
\vspace{-0.35cm}
\begin{figure}[hbt]
  \centering
  \includegraphics[width=0.7\linewidth]{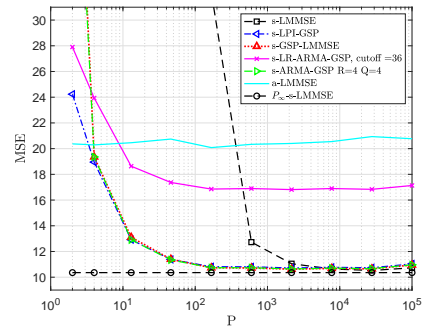}
  \caption{The MSE versus $P$ for the different estimators, where $\sigma^2 = 0.05$ and $\beta = 3$.}
  \label{fig:Diff N}
  \vspace{-0.5cm}
\end{figure}

In Fig. \ref{fig:Diff sigma}, we present the MSE versus the noise variance, $\sigma^2$ for $P=500$ and $\beta=3$. It can be seen that the MSE of all estimators (except the aLMMSE from \eqref{LS}, which is based on a linearization of the model) increases as the noise variance increases. 
In this case, the sample-GSP-LMMSE, the sample linear pseudo-inverse GSP, and the sample ARMA GSP estimators outperform the sample-LMMSE estimator and approach the lower bound obtained by  $P_\infty$-LMMSE for all values of the noise variance.
\vspace{-0.25cm}
\begin{figure}[hbt]
  \centering
  \includegraphics[width=0.7\linewidth]{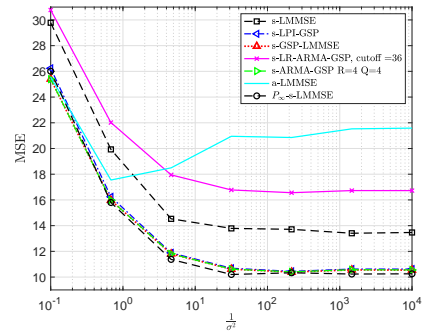}
  \caption{The MSE versus $\frac{1}{\sigma^2}$ for the different estimators, where $P = 500$ and $\beta = 3$.}
  \label{fig:Diff sigma}
  \vspace{-0.25cm}
\end{figure}

In order to demonstrate the 
complexity of the  estimators empirically,
 the average computation time was evaluated  using Matlab on an Intel Core(TM) i7-7700K CPU computer, 4.2 GHz. In Fig. \ref{fig:Times}, we present the runtime of  the different estimators versus the target MSE, where $\sigma^2 = 0.05$ and $\beta = 3$. 
It can be seen that the runtime of all the estimators increases as the target MSE decreases. 
The runtimes of the sample ARMA GSP and the sample LR-ARMA GSP  estimators are the highest since finding their optimal coefficients  requires  solving a nonconvex optimization problem   (see Subsections \ref{ARMA_filter_subsection} and \ref{LC_ARMA_filter_subsection}, respectively), while the LR-ARMA has a lower runtime since it has  fewer coefficients.
The proposed sample-GSP-LMMSE estimator has the lowest runtime for any target MSE, and the proposed sample linear pseudo-inverse GSP estimator is the second-best in terms of runtime.

It should be noted that the topology is stationary in this simulation. Thus, the EVD of the Laplacian matrix, $\Lmat$, is assumed to be known and given in advance.
When there is a change in the network, the 
computational complexity of updating the sample-h-GSP estimators is much lower than those of the other estimators and, thus, has a shorter runtime.
This is since the sample-h-GSP estimators use the graph filter coefficients, $\hat{\alphavec}^{\text{sample}}$, that were evaluated on the initial topology. In contrast,  the sample-LMMSE and the sample-GSP-LMMSE estimators require a reevaluation  for each new topology, as shown in the following subsection.
  \vspace{-0.4cm}
\begin{figure}[hbt]
  \centering
  \includegraphics[width=0.75\linewidth]{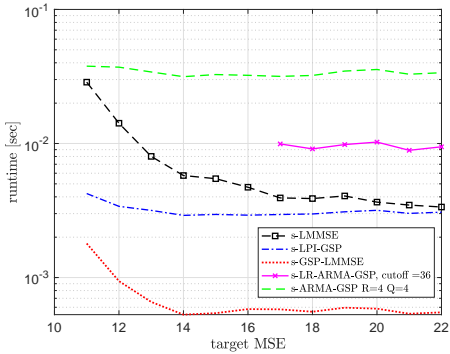}
  \caption{The runtime for the different estimators versus the target MSE, where $\sigma^2 = 0.05$ and $\beta = 3$.}
  \label{fig:Times}
  \vspace{-0.5cm}
\end{figure}

\subsection{Example B: estimation under topology changes}
\label{subsection_change}
\vspace{-0.1cm}
In this subsection, we discuss the case where the underlying topology changes over time. 
For example, when a sensor fails or changes its location, the sensor network's topology changes.
Similarly, the power grid topology may be changed by failure, opening and closing of switches on power lines, and the presence of new loads and generators.
When $\Lmat$ changes, the measurement function, $\gvec(\Lmat,\xvec)$, from the model in \eqref{Model} and the distribution from \eqref{theta_distribution_simulation}, change.
Our goal is to estimate $\xvec$ without generating new dataset.
The MSE of the different estimators is considered under random changes in the topology with the constraints that the graph will remain well-connected, i.e., assuming that $\lambda_2$, which is related to the connectivity \cite{Newman_2010} did not reduce significantly. 
The MSE shown is the average MSE over 100 random changes on the graph.

\subsubsection{Estimation under edge changes}
In this case, the changes in the topology are due to the addition or removal of edges. Thus, the problem dimension did not change.
In order to evaluate all the estimators from Section \ref{Methods} we use the mean of $\xvec$, i.e., $\EX[\xvec]$, the eigenvalue and eigenvectors of the Laplacian matrix with the historical sample values, such as $\hat{\Cmat}_{\yvec\yvec}$ from \eqref{hat_C_y}, $\hat{\Dmat}_{\tilde{\yvec}\tilde{\yvec}}$ from \eqref{sample_D}.
 It should be noted that in the following, the sample-LMMSE estimator is based on the initial topology, where the sample-GSP-LMMSE, the sample linear pseudo-inverse GSP,  and the sample ARMA GSP estimators have been updated to the new topology, as described after \eqref{derivative_opt_alpha}.

Figures \ref{fig:frequency_response} and \ref{fig:frequency_response_7_new} present the graph frequency response  of the different sample-GSP  estimators, where $P = 500$, $\sigma^2 = 0.05$ and $\beta = 3$ for the stationary network from Example A (Fig.  \ref{fig:frequency_response}) and for the topology change from Example B, where   $M = 7$ new edges were added  (Fig. \ref{fig:frequency_response_7_new}).
It can be seen that in both cases, all  GSP filters
achieve almost the same graph frequency response as  the sample-GSP-LMMSE estimator and, thus, they can be considered as robust to topology changes.
The graph frequency response of the sample GSP-LMMSE at $\lambda_1 $ is $0$, which is approximated by the h-GSP filters as a small value.
However, it can be seen that the graph frequency response is a nonzero (small) response (i.e. the output $\yvec$ is not a perfect graph low-frequency signal) while the graph frequency response of the LR-ARMA GSP estimator is absolutely zero for $\lambda>36$. Therefore, the LR-ARMA GSP estimator  does not perform well in the simulations.
For perfectly low-frequency signals (not shown here due to space limitations), the LR-ARMA GSP estimator achieves the same performance as the other h-GSP estimators.
In addition, it can be seen in Fig. \ref{fig:frequency_response_7_new} that the graph frequency response of the sample-GSP-LMMSE estimator that  was evaluated based on the initial topology (red)  is a less accurate approximation of the optimal graph frequency response. Finally, 
the desired graph frequency response, $\hat{f}(\Lambdamat)$, includes a sharp transition since $f(\lambda_1)=0$; thus, a linear graph filter as in  \eqref{Linear_filter_def} is inappropriate for this case  since it requires a high filter order 
leading to a high implementation cost and limited accuracy \cite{sparse_paper}.  

\begin{figure}[hbt]
    \centering
	\subcaptionbox{\label{fig:frequency_response}}[\linewidth]
	{ \includegraphics[width=0.75\linewidth]{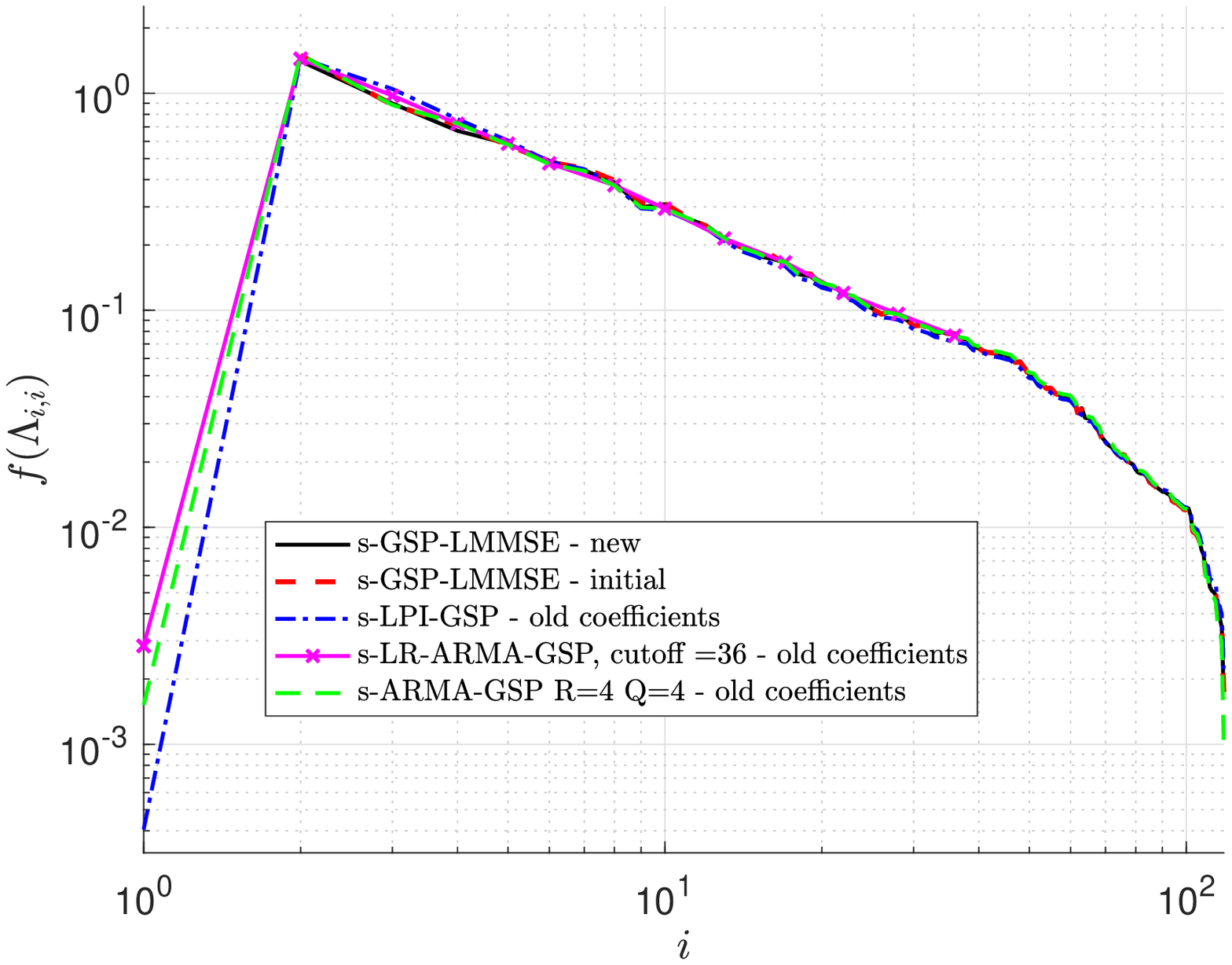}\vspace{-0.2cm}}
	\subcaptionbox{\label{fig:frequency_response_7_new}}[\linewidth]
	{\includegraphics[width=0.75\linewidth]{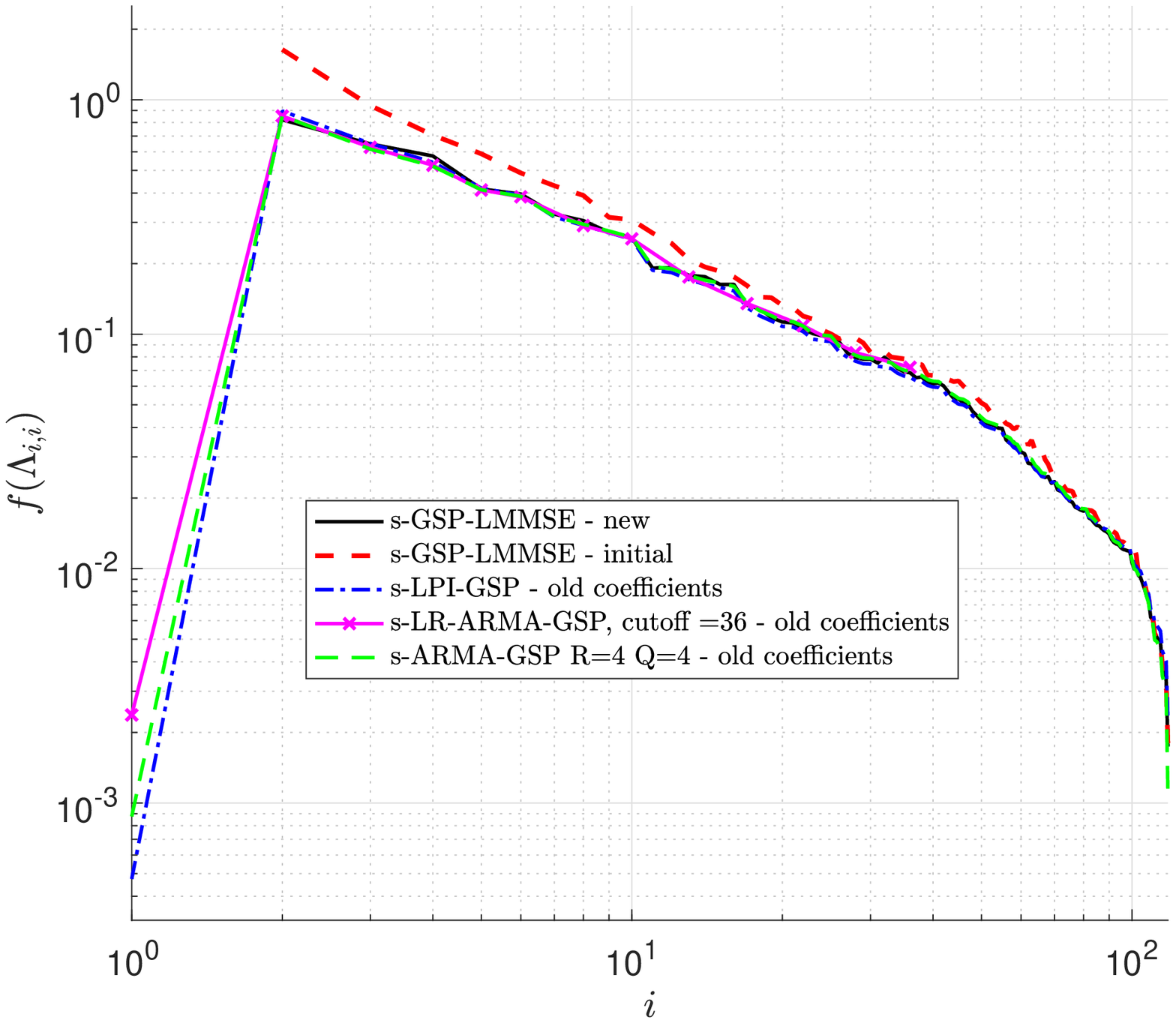}\vspace{-0.2cm}} 
	\caption{
 The graph frequency response of the different estimators, where $P = 500$, $\sigma^2 = 0.05$ and $\beta = 3$, for (a) a stationary network; and
	(b) a network with the addition of $M=7$ edges. }
\end{figure}

Figures \ref{fig:add edge} and \ref{fig:remove edge} present the case where the  sample-mean values were calculated from the dataset which is  evaluated on a topology before $M$ edges are added or removed.
Since there is no straightforward methodology to update the sample-LMMSE estimator to the new topology, its performance in the sense of the MSE is impaired for both cases: added and removed edges. 
Moreover, even when a small number of new edges is added, the MSE of the sample-LMMSE significantly increases.
It can be seen that even for $2$ new edges, the sample LR-ARMA GSP estimator has a lower MSE than the sample-LMMSE estimator.
The sample linear pseudo-inverse GSP  and the sample ARMA GSP estimators
have a lower MSE for any number of edges added or removed.
The sample-GSP-LMMSE estimator that has been updated to the new topology, has the same MSE as the sample linear pseudo-inverse GSP estimator and the sample ARMA GSP estimator for a small number of edges added or removed, and a slightly higher MSE for a large number of edges added or removed.
The aMMSE estimator performance increases for edges removed and improves for the case where new edges are added.
It should be noted that in addition to their advantage in terms of MSE, the computational complexity of updating the sample-h-GSP estimators is  lower than those of the sample-LMMSE and the sample-GSP-LMMSE estimators. 
This is
since the sample-h-GSP estimators use the graph filter coefficients, $\hat{\alphavec}^{\text{sample}}$, that were evaluated on the initial topology, while the sample-LMMSE and the sample-GSP-LMMSE estimators 
require to reevaluate the  estimators from scratch for the new topology.
\begin{figure}[hbt]
    \centering
	\subcaptionbox{\label{fig:add edge}}[\linewidth]
	{ \includegraphics[width=0.75\linewidth]{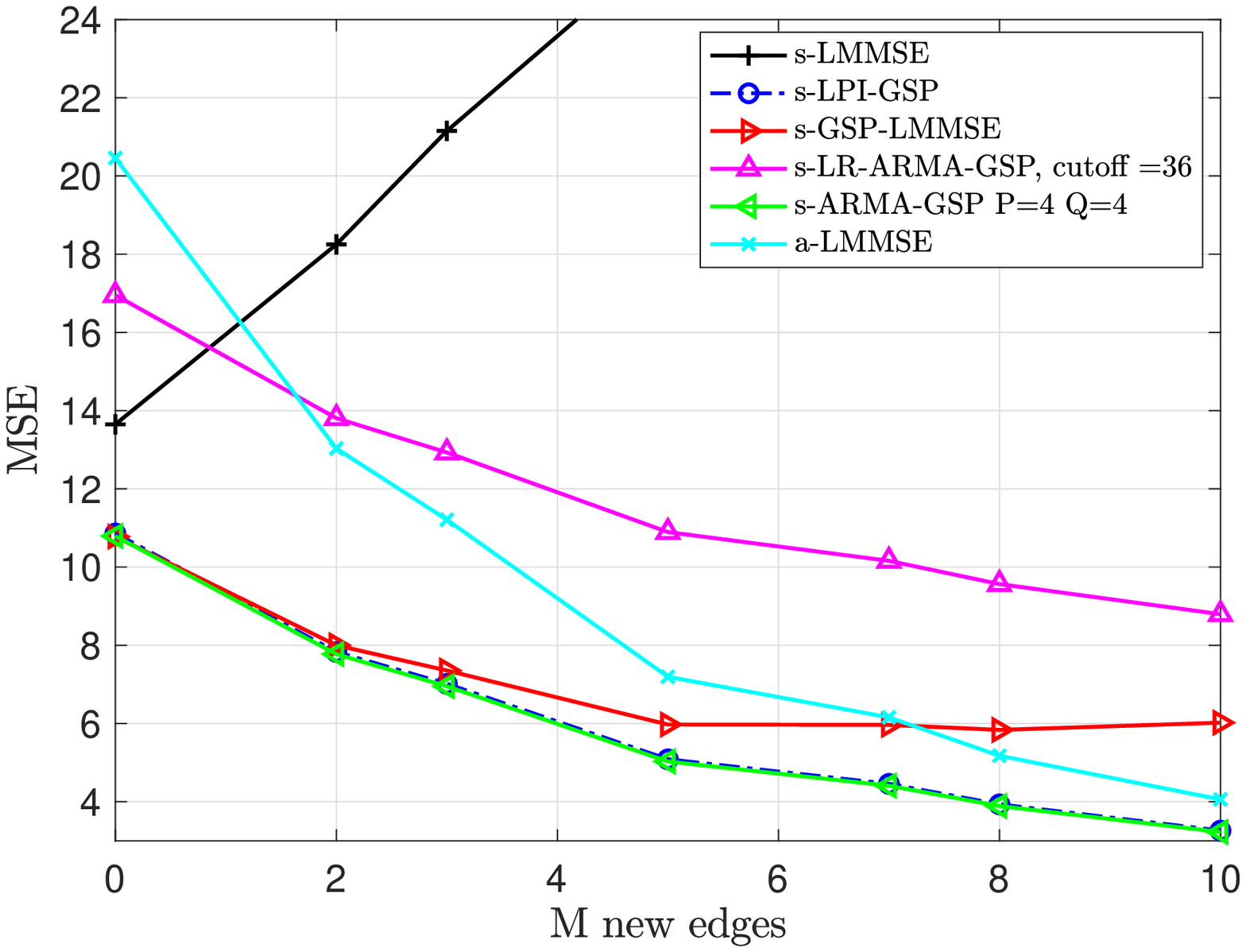}\vspace{-0.2cm}}
	\subcaptionbox{\label{fig:remove edge}}[\linewidth]
	{\includegraphics[width=0.75\linewidth]{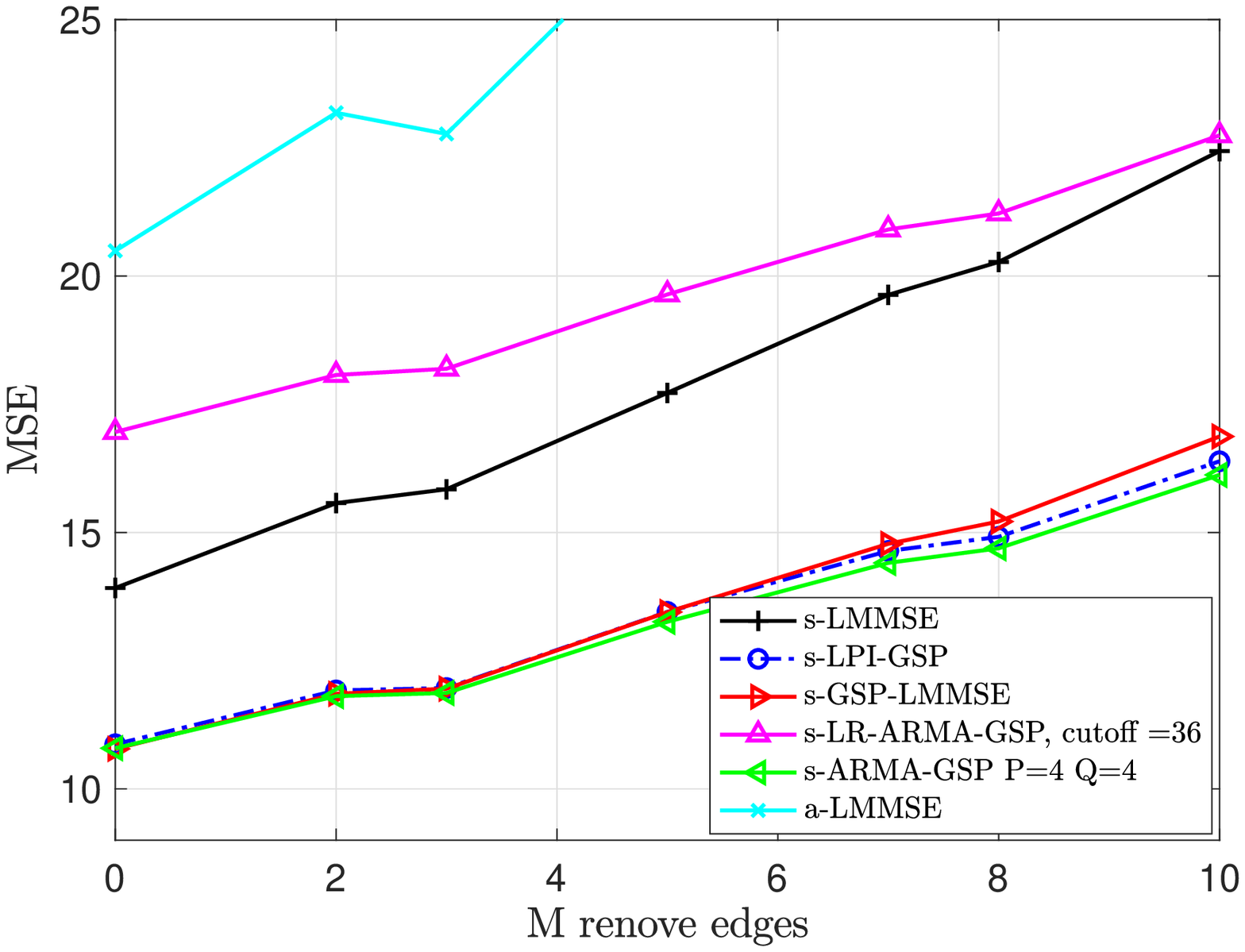}\vspace{-0.2cm}} 
	\caption{
	The MSE of the updated GSP estimators, the sample-LMMSE estimator, and the aLMMSE estimator, where $P = 500$, $\sigma^2 = 0.05$, and $\beta = 3$, for (a) the addition of $M$ new edges; and
	(b) a removal of $M$ edges. The error bars show conﬁdence intervals of $\pm 0.5$ standard deviations.}
\end{figure}

\subsubsection{Estimation under vertices changes}
When $M$ vertices are removed or added, the problem dimension changes.
That is, $\xvec,\yvec \in \mathbb{R}^{N\pm M}.$
Since the sample-LMMSE and the sample-GSP-LMMSE estimators are $\mathbb{R}^{N} \rightarrow \mathbb{R}^{N}$ estimators, they cannot be implemented in the new problem.
Therefore, we use the following methods:
1) for $M$ new vertices, the sample-LMMSE and the sample-GSP-LMMSE estimators estimate the signal at the  new vertices by zero and do not use measurements from those vertices; 2) for $M$ removed vertices, the sample-LMMSE and the sample-GSP-LMMSE estimators are updated by removing the appropriate rows and columns. That is, 
 the sample-LMMSE estimator from \eqref{Monte_Carlo_LMMSE} is given by
\begin{equation}
\label{update1}
\begin{split}
    \hat\xvec^{(\text{sLMMSE})}_{\mathcal{S}} = \EX[\xvec] + [\hat{\Cmat}_{{\xvec\yvec}}\hat{\Cmat}_{{\yvec\yvec}}^{-1}]_{\mathcal{S}}\left(\yvec_{\mathcal{S}} - \hat{\yvec}_{\mathcal{S}}\right),
    \end{split}
\end{equation}
and the sample-GSP-LMMSE estimator from \eqref{sample_LMMSE_GSP} is given by
\begin{equation}
\label{update2}
\begin{split}
    \hat{\xvec}^{(\text{sGSP-LMMSE})}_{\mathcal{S}}  = \EX[\xvec] +  [\Vmat\hat{f}(\Lambdamat)\Vmat^T]_{\mathcal{S}}
    (\yvec_{\mathcal{S}} - \hat{\yvec}_{\mathcal{S}}),
    \end{split}
\end{equation}
where we use 
the historical eigenvalue and eigenvectors of the Laplacian matrix and the historical sample values of the covariance matrices, evaluated using the historical dataset.
In \eqref{update1} and \eqref{update2}
we use the notation that
 $\Amat_{\mathcal{S}}$ is the submatrix of $\Amat$ whose rows and columns are indicated by the set $\mathcal{S}$,
 where $\mathcal{S}$
is the set of the remaining vertices. 
In addition, any sample-h-GSP estimator from \eqref{sGSP_estimator}, can be updated to the new topology for both cases: vertices added or removed, as follows:
\begin{equation} \label{update_GSP_estimator2}
    \hat{\xvec}^{(\text{update-sh-GSP})}  = \EX[\xvec] +  \Vmat h(\Lambdamat;\hat{\alphavec}^{\text{sample}})\Vmat^T(\yvec - \hat{\bar{\yvec}}),
\end{equation}
where we use the true value of $\EX[\xvec]$, the eigenvalue and eigenvectors of the Laplacian matrix with the filter coefficient, $\hat{\alphavec}^{\text{sample}}$, that were evaluated using the historical dataset and $\hat{\bar{\yvec}} \in \mathbb{R}^{N\pm M} $ that is defined by 
\begin{equation}
    [\hat{\bar{\yvec}}]_n  =\begin{cases}
    0 & \text{if } n \text{ is added} \\
    [\hat{\yvec}]_n, &\text{if  } n \text{ is unchanged} \end{cases},~n=1,\ldots,N,
    \end{equation}
 where $\hat{\yvec}$  is evaluated for the old topology by  \eqref{sample_mean}.

Figures  \ref{fig:add vertices} and  \ref{fig:remove vertices} present the case where the historical sample values were calculated from the dataset evaluated on a topology before the addition or removal of the  $M$ vertices, respectively.
Since there is no straightforward methodology to update the sample-LMMSE and the sample-GSP-LMMSE estimators to the new topology, their  MSE is impaired for added and removed vertices. 
Moreover, even when a small number of new vertices is added, the MSE of the sample-LMMSE and the sample-GSP-LMMSE estimators significantly increases.
The other estimators' performance for the cases of added or removed vertices is similar to the result in the cases of added or removed edges in Figs. \ref{fig:add edge} and \ref{fig:remove edge}.
In Figs. \ref{fig:add vertices} and \ref{fig:remove vertices}, the conﬁdence intervals of $\pm 0.5$ standard deviations are significant due to the variability in the topology of the different experiments. Thus, increasing the number of experiments does not reduce this conﬁdence interval.
\begin{figure}[hbt]
    \centering
	\subcaptionbox{\label{fig:add vertices}}[\linewidth]
	{ \includegraphics[width=0.75\linewidth]{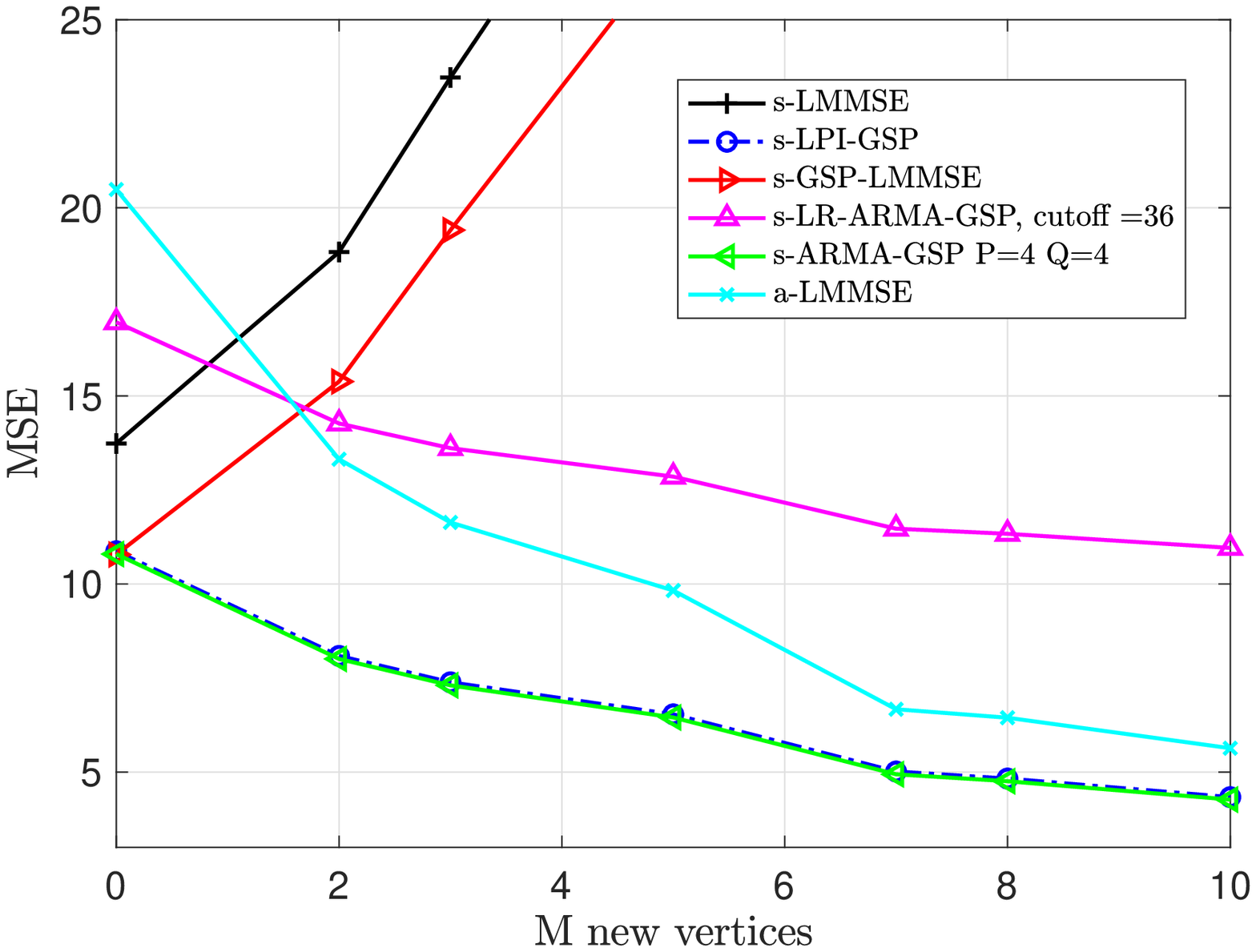}\vspace{-0.2cm}}
	\subcaptionbox{\label{fig:remove vertices}}[\linewidth]
	{\includegraphics[width=0.75\linewidth]{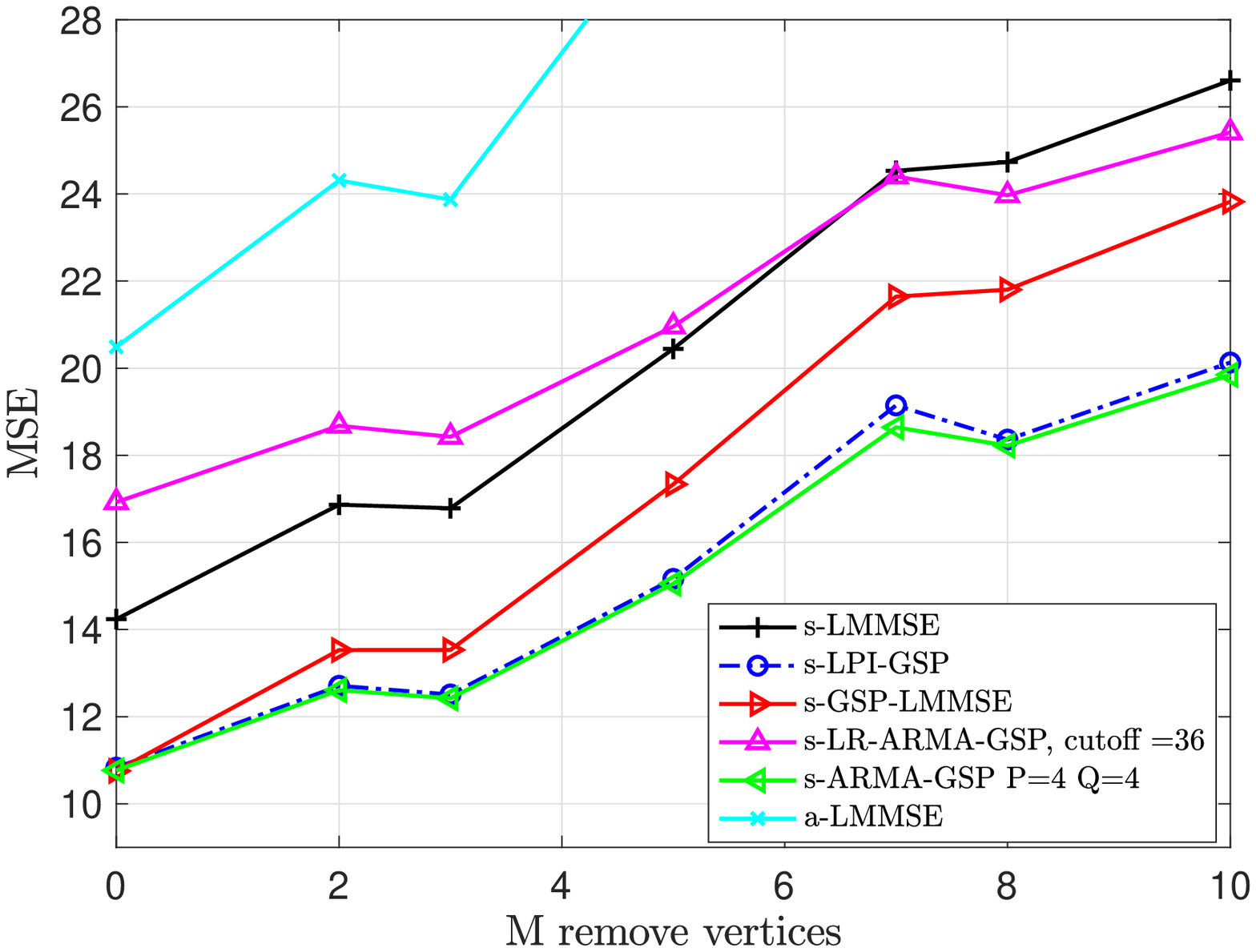}\vspace{-0.2cm}} 
	\caption{
	The MSE of the updated GSP estimators, the sample-LMMSE estimator, and the aLMMSE estimator, where $P = 500$, $\sigma^2 = 0.05$, and $\beta = 3$, for  (a) an addition of $M$ new vertices; and
	(b) a removal of  $M$. The error bars show conﬁdence intervals of $\pm 0.5$ standard deviations.}
	\vspace{-0.5cm}
\end{figure}
\section{conclusion} \label{conclusion}
In this paper, we discuss a GSP-based Bayesian approach for the recovery of random graph signals from nonlinear measurements. We develop the GSP-LMMSE estimator, which minimizes the MSE among the subset of estimators that are represented as an output of a graph filter. We evaluate the conditions for the GSP-LMMSE estimator to coincide with
the LMMSE estimator and with the graphical Wiener filter. If the distributions of the graph signal and the observations are intractable, the sample-mean versions of the different estimators can be used. The diagonal structure of the sample-GSP-LMMSE estimator in the graph frequency domain bypasses the requirement for  an extensive dataset  to obtain stable estimation of the sample-LMMSE estimator. However, the GSP-LMMSE estimator is a function of the specific graph structure with fixed dimensions, and thus it is not necessarily optimal when the topology changes and  is not adaptive to changes in the number of vertices. Therefore, we develop the sample-h-GSP estimators that are the MSE-optimal parametrization of the sample-GSP-LMMSE estimator by  graph filters. The sample-h-GSP estimators can be updated when the topology changes without generating a new  dataset, even in the case of changes in the number of vertices.

In the simulations, we show that the proposed sample-GSP estimators achieve lower MSE than the sample-LMMSE estimator for a limited training dataset, and they coincide with the sample-LMMSE estimator for sufficiently large datasets.
In addition, it is shown that 
the three specific parametric implementations of the GSP-LMMSE: the linear pseudo-inverse GSP estimator,  ARMA GSP estimator, and the low-rank ARMA GSP estimator, are  robust to
changes in the topology  without the need for generating new training data.
The sample linear pseudo-inverse GSP
 and the sample ARMA GSP estimators achieve the lowest MSE in these cases, where the ARMA GSP estimator requires less filter coefficients. 
 Thus, the proposed approach  is a practical method to recover nonlinear graph signals in networks.
 
There are several directions left for future work. 
One direction is to study the use of graph neural networks and other nonlinear approaches \cite{Isufi_Ribeiro,gama2020graphs}.
In addition,  the development of Bayesian  bounds on the MSE of general (not necessarily linear) estimators
of graph signals, in a similar manner to the  non-Bayesian  graph Cram$\acute{\text{e}}$r-Rao bound from \cite{routtenberg2020}, should be investigated. Finally, it is interesting to consider distributed implementation of the proposed estimators that include the computation of the optimal coefficient vector and the diagonal sample covariance matrices.

\appendices
	\renewcommand{\thesectiondis}[2]{\Alph{section}:}
\section{Proof of Theorem \ref{claim_separately_Model}} \label{separately_Model_Appendix}
In this Appendix, we show that under the conditions of Theorem \ref{claim_separately_Model} the equality in \eqref{coincides_condition} holds.

Since $\xvec$ and $\wvec$ are statistically independent under the considered model, the covariance matrix of $\tilde{\yvec}$ is given by
\begin{equation} \label{separately_Model_Appendix_covariance_matrix}
  \Cmat_{\tilde{\yvec}\tilde{\yvec}} 
  = \Cmat_{\tilde{\gvec}\tilde{\gvec}} + \Cmat_{\tilde{\wvec}\tilde{\wvec}},
\end{equation}
where $\tilde{\gvec}=\Vmat^T\gvec(\Lmat,\xvec)$.
Similarly,
\begin{equation} \label{crros_covariance_matrix_app_2}
\Cmat_{\tilde{\xvec}\tilde{\yvec}}= \Cmat_{\tilde{\xvec}\tilde{\gvec}}.
\end{equation}
Since from Condition \ref{cond1} the measurement function, $\gvec(\Lmat,\xvec)$, satisfies \eqref{separately_Model}, we obtain that 
the off-diagonal elements of the  matrix $\Cmat_{\tilde{\gvec}\tilde{\gvec}}$ from \eqref{separately_Model_Appendix_covariance_matrix} satisfy
\begin{eqnarray} \label{covariance_matrix_app_1}
    [\Cmat_{\tilde{\gvec}\tilde{\gvec}}]_{n,k}  = 
    \EX\Big[([\tilde{\gvec}(\Lmat,\xvec)]_{n} - \EX[\tilde{\gvec}(\Lmat,\xvec)]_n)
    \hspace{1.5cm}\nonumber\\ \times ([\tilde{\gvec}(\Lmat,\xvec)]_{k} - \EX[\tilde{\gvec}(\Lmat,\xvec)]_k)\Big] \hspace{1.45cm}\nonumber\\
    =
    \EX\Big[([\tilde{\gvec}(\Lmat,\tilde{x}_n\vvec_n)]_{n} - \EX[\tilde{\gvec}(\Lmat,\xvec)]_n)
    \hspace{0.9cm}\nonumber\\ \times ([\tilde{\gvec}(\Lmat,\tilde{x}_k\vvec_k)]_{k} - \EX[\tilde{\gvec}(\Lmat,\xvec)]_k)\Big] 
    =0,
\eeqna
for any $n\neq k$,
where the last equality follows from Condition \ref{cond2}. Substituting \eqref{covariance_matrix_app_1} in \eqref{separately_Model_Appendix_covariance_matrix} and using Condition \ref{cond3}, which implies that $\Cmat_{\tilde{\wvec}\tilde{\wvec}}$ is a diagonal matrix, the matrix $\Cmat_{\tilde{\yvec}\tilde{\yvec}}$ is also a diagonal matrix, which satisfies
\begin{equation}\label{separately_Model_Appendix_C_xx_f}
    \Cmat_{\tilde{\yvec}\tilde{\yvec}} =	\text{diag}\left(  \text{diag}(\Cmat_{\tilde{\yvec}\tilde{\yvec}} ) \right)= {\Dmat}_{\tilde{\yvec}\tilde{\yvec}}.
\end{equation}
Similarly, using  Condition \ref{cond1}, the off-diagonal elements of the cross-covariance matrix,
$\Cmat_{\tilde{\xvec}\tilde{\gvec}}$ from \eqref{crros_covariance_matrix_app_2} satisfy
\begin{eqnarray} \label{separately_Model_Appendix_C_tx}
    [\Cmat_{\tilde{\xvec}\tilde{\gvec}}]_{n,k}  
    = 
    \EX\left[\left([\tilde{\xvec}]_n- \EX[\tilde{\xvec}]_n\right)
    \left([\tilde{\gvec}(\Lmat,\xvec)]_k- \EX[\tilde{\gvec}(\Lmat,\xvec)]_k\right)\right] 
     \nonumber\\ = 
    \EX\left[\left([\tilde{\xvec}]_n- \EX[\tilde{\xvec}]_n\right) 
    \left([\tilde{\gvec}(\Lmat,\tilde{x}_k\vvec_k)]_k- \EX[\tilde{\gvec}(\Lmat,\xvec)]_k\right)\right] 
    \hspace{-0.55cm}\nonumber\\= 0,\hspace{6.33cm}
\end{eqnarray}
for any $n\neq k$, where the last equality follows from Condition \ref{cond2}. 
By substituting \eqref{separately_Model_Appendix_C_tx} in \eqref{crros_covariance_matrix_app_2}, we have
\begin{equation} \label{separately_Model_Appendix_C_tx_f}
    \Cmat_{\tilde{\xvec}\tilde{\yvec}} ={\text{diag}}(\Cmat_{\tilde{\xvec}\tilde{\yvec}})= \text{diag}({\dvec}_{\tilde{\xvec}\tilde{\yvec}}).
\end{equation}
Therefore, \eqref{separately_Model_Appendix_C_xx_f} and \eqref{separately_Model_Appendix_C_tx_f} imply that 
 $\Cmat_{\tilde{\xvec}\tilde{\yvec}}$ and $\Cmat_{\tilde{\yvec}\tilde{\yvec}}$ are diagonal matrices, and that \eqref{coincides_condition} holds.

\section{Proof of Theorem \ref{claim_graphical_Model}} \label{graphical_Model_Appendix}
In this Appendix, we show that under the assumption of Theorem \ref{claim_graphical_Model} the equality in \eqref{coincides_condition} holds.
By using \eqref{g_filter} from Condition \ref{cond4}, we obtain that
$\tilde{\gvec}(\Lmat,\xvec) =  f(\Lambdamat) \Vmat^{T}{\xvec}$ and thus, 
\begin{eqnarray} \label{covariance_matrix_app_2}
    \Cmat_{\tilde{\gvec}\tilde{\gvec}}
    =
    \EX[( f(\Lambdamat)(\tilde{\xvec}- \EX[\tilde{\xvec})])
    ( f(\Lambdamat)(\tilde{\xvec}- \EX[\tilde{\xvec})])^T]
\nonumber\\
    =
    f(\Lambdamat)\Cmat_{\tilde{\xvec}\tilde{\xvec}}f(\Lambdamat),\hspace{3.5cm}
\end{eqnarray}
where we use the symmetry of $f(\Lambdamat)$.
By substituting \eqref{covariance_matrix_app_2} in \eqref{separately_Model_Appendix_covariance_matrix} from Appendix \ref{separately_Model_Appendix} and using Condition \ref{cond3} and Condition \ref{cond5}, which implies that $\Cmat_{\tilde{\wvec}\tilde{\wvec}}$ and $\Cmat_{\tilde{\xvec}\tilde{\xvec}}$ are diagonal matrices, we obtain that  $\Cmat_{\tilde{\yvec}\tilde{\yvec}}$ is a diagonal matrix, which satisfies
\begin{equation}\label{graphical_Model_Appendix_C_xx_f}
    \Cmat_{\tilde{\yvec}\tilde{\yvec}} =\text{diag}\left(  \text{diag}(\Cmat_{\tilde{\yvec}\tilde{\yvec}} ) \right)= {\Dmat}_{\tilde{\yvec}\tilde{\yvec}}.
\end{equation}
Similarly, the cross-covariance matrix, $\Cmat_{\tilde{\xvec}\tilde{\gvec}}$ from \eqref{crros_covariance_matrix_app_2} satisfies
\begin{eqnarray} \label{graphical_Model_Appendix_C_tx}
    \Cmat_{\tilde{\xvec}\tilde{\gvec}}  = 
    \EX\left[(\tilde{\xvec}- \EX[\tilde{\xvec}])
    ( f(\Lambdamat)\left(\tilde{\xvec}- \EX[\tilde{\xvec}]\right))^T\right]
=
    \Cmat_{\tilde{\xvec}\tilde{\xvec}}f(\Lambdamat),
\end{eqnarray}
where we use the fact that $f(\Lambdamat)$ is a symmetric matrix. 
By substituting \eqref{graphical_Model_Appendix_C_tx} in \eqref{crros_covariance_matrix_app_2} from Appendix \ref{separately_Model_Appendix}, the cross-covariance matrix of $\tilde{\xvec}$ and $\tilde{\yvec}$ is a diagonal matrix, and satisfies
\begin{equation} \label{graphical_Model_Appendix_C_tx_f}
    \Cmat_{\tilde{\xvec}\tilde{\yvec}} = \text{diag}({\dvec}_{\tilde{\xvec}\tilde{\yvec}}),
\end{equation}
 where ${\dvec}_{\tilde{\xvec}\tilde{\yvec}}$ is defined in \eqref{d_def}.
 Therefore, \eqref{graphical_Model_Appendix_C_xx_f} and \eqref{graphical_Model_Appendix_C_tx_f} imply that 
 $\Cmat_{\tilde{\xvec}\tilde{\yvec}}$ and $\Cmat_{\tilde{\yvec}\tilde{\yvec}}$ are diagonal matrices, and that \eqref{coincides_condition} holds.

\section{Derivation of \eqref{filter_coefficients_opt}} \label{WLS_formulation}
In this Appendix we show that solving \eqref{derivative_opt_alpha} is equivalent to \eqref{filter_coefficients_opt}.
By adding and subtracting $\hat{f}(\Lambdamat)(\tilde{\yvec} - \EX[\tilde{\yvec}])$ from the r.h.s. of \eqref{derivative_opt_alpha}, one obtains
\begin{eqnarray}
  \hat{\alphavec} 
  = \argmin_{\alphavecsmall\in \Omega_{\alphavecsmall}} 
  \EX [||h(\Lambdamat;\alphavec)(\tilde{\yvec} - \EX[\tilde{\yvec}]) - (\tilde{\xvec} - \EX[\tilde{\xvec}]) 
  \hspace{0.9cm}\nonumber \\ 
  -\hat{f}(\Lambdamat)(\tilde{\yvec} - \EX[\tilde{\yvec}]) +
  \hat{f}(\Lambdamat)(\tilde{\yvec} - \EX[\tilde{\yvec}]) ||^2]
  \hspace{0.40cm}\nonumber\\=
  \argmin_{\alphavecsmall\in \Omega_{\alphavecsmall}} 
  \EX [||
  ( (h(\Lambdamat;\alphavec)-\hat{f}(\Lambdamat) ) (\tilde{\yvec} - \EX[\tilde{\yvec}]))
  ||^2] 
  \hspace{0.13cm}\nonumber\\+
  2\EX [(\hat{f}(\Lambdamat)(\tilde{\yvec} - \EX[\tilde{\yvec}]) - ((\tilde{\xvec} -\EX[\tilde{\xvec}]) ))^T
  \hspace{0.009cm}\nonumber\\ \times
  ( (h(\Lambdamat;\alphavec)-\hat{f}(\Lambdamat) ) (\tilde{\yvec} - \EX[\tilde{\yvec}])
  )],\hspace{0.85cm}
\label{last_eq}
\end{eqnarray}
where the last equality is obtained by removing constant terms w.r.t. $\alphavec$.
In addition, by substituting $\hat{f}(\Lambdamat)$ from \eqref{opt_f} in the last term of \eqref{last_eq} it can be verified that
\begin{eqnarray}
  \EX [(\hat{f}(\Lambdamat)(\tilde{\yvec} - \EX[\tilde{\yvec}]) - ((\tilde{\xvec} -\EX[\tilde{\xvec}]) ))^T
  \hspace{2.3cm}\nonumber\\
  \times
  ((h(\Lambdamat;\alphavec)-\hat{f}(\Lambdamat)) (\tilde{\yvec} - \EX[\tilde{\yvec}])
  )] 
  \nonumber\\=
  (\text{diag}(\hat{f}(\Lambdamat)))^T\Dmat_{\tilde{\yvec}\tilde{\yvec}}( \text{diag}(h(\Lambdamat;\alphavec))- \text{diag}(\hat{f}(\Lambdamat))  )
  \hspace{0.06cm}\nonumber\\
  - \dvec_{\tilde{\xvec}\tilde{\yvec}}^T( \text{diag}(h(\Lambdamat;\alphavec))- \text{diag}(\hat{f}(\Lambdamat))  ) = 0,\hspace{1.5cm}
\label{zerp_eq}
\end{eqnarray}
where $\dvec_{\tilde{\xvec}\tilde{\yvec}}$
 and $\Dmat_{\tilde{\yvec}\tilde{\yvec}}$ are defined in  \eqref{d_def}.
By substituting 
\eqref{zerp_eq} in \eqref{last_eq}, one obtains
\begin{eqnarray}
  \hat{\alphavec} = \argmin_{\alphavecsmall\in \Omega_{\alphavecsmall}}
  \EX [||
   (h(\Lambdamat;\alphavec)-\hat{f}(\Lambdamat) ) (\tilde{\yvec} - \EX[\tilde{\yvec}])
  ||^2] 
  \nonumber\\
    = 
  \EX [||
   \text{diag}(\tilde{\yvec} - \EX[\tilde{\yvec}])\text{diag}(h(\Lambdamat;\alphavec)-\hat{f}(\Lambdamat) ) 
  ||^2]  
   \hspace{-0.138cm}\nonumber\\ = 
    (\text{diag}(h(\Lambdamat;\alphavec))-\text{diag}(\hat{f}(\Lambdamat)) )^T\Dmat_{\tilde{\yvec}\tilde{\yvec}}
    \hspace{1.4cm}\nonumber\\ 
    \times (\text{diag}(h(\Lambdamat;\alphavec))-\text{diag}(\hat{f}(\Lambdamat)) ),
\label{last_app_B}
\end{eqnarray} 
where the second equality follows since $h(\Lambdamat;\alphavec))$ and $(\hat{f}(\Lambdamat)$ are diagonal matrices.
Then, since $\Dmat_{\tilde{\yvec}\tilde{\yvec}}$ is a diagonal matrix, and thus, a symmetric matrix and it is assumed that it is also a non-singular (and therefore, positive definite) matrix, we obtain \eqref{filter_coefficients_opt} by substituting $\Dmat_{\tilde{\yvec}\tilde{\yvec}} = \Dmat_{\tilde{\yvec}\tilde{\yvec}}^{\frac{1}{2}}\Dmat_{\tilde{\yvec}\tilde{\yvec}}^{\frac{1}{2}}$ and ${\text{diag}}({\text{diag}}(\dvec_{\tilde{\xvec}\tilde{\yvec}})\Dmat_{\tilde{\yvec}\tilde{\yvec}}^{-1}) = \Dmat_{\tilde{\yvec}\tilde{\yvec}}^{-1}\dvec_{\tilde{\xvec}\tilde{\yvec}} $ in \eqref{last_app_B}.

\section{Proof of Claim \ref{claim1} and Claim \ref{claim2}} \label{rank_Appendix}
We prove Claim \ref{claim1} by showing  that $\text{rank}(\bar\Gammamat) = K+1$ under the condition in Claim \ref{claim1}, where
$\bar\Gammamat_K$ is defined in \eqref{barGammamat} and $K+1\leq N$.
First, it can be seen that 
\begin{eqnarray} \label{Gamma2}
\bar{\Gammamat}_K 
	=
	\text{diag}(1,\lambda_2^{-K},\dots,\lambda_N^{-K}) \Bmat_K(\lambda_1,\ldots,\lambda_N),
\end{eqnarray}
where 
\[
\Bmat_K(\lambda_1,\ldots,\lambda_N)\define 	\begin{bmatrix}
		1 & 0 &\dots & 0&0 \\
		\lambda_2^{K} & \lambda_2^{K-1} &\dots &\lambda_2& 1 \\
		\vdots & \vdots & \ddots & \vdots& \vdots\\
		\lambda_N^{K} & \lambda_N^{K-1} & \dots &\lambda_N& 1
	\end{bmatrix}.
	\]
We assume that the graph is connected, i.e. $0<\lambda_n$, $ n=2,\ldots,N$.
Thus, $\text{diag}(1,\lambda_1^{-K},\dots,\lambda_N^{-K})$ is a non-singular matrix.
The multiplication of $\Bmat_K(\lambda_1,\ldots,\lambda_N)$ by an $N\times N$ non-singular matrix, 
in \eqref{Gamma2} implies that (see 0.4.6 in \cite{Horn_Johnson_book})
\begin{eqnarray}
\label{Gamma3}
	\text{rank}(\bar{\Gammamat}_K )
	= \text{rank}(\Bmat_K(\lambda_1,\ldots,\lambda_N)).
\end{eqnarray}
By reordering the columns of $\Bmat_K(\lambda_1,\ldots,\lambda_N)$ and using the properties  of 
the Vandermonde matrix $\Phivec(N-1,K)$
from \eqref{Gammamat_def} (see 0.9.11 \cite{Horn_Johnson_book}), we obtain that if there are $K+1$ distinct eigenvalues, then 
$\text{rank}(\Bmat_K(\lambda_1,\ldots,\lambda_N))=K+1$.
By using \eqref{Gamma3}, this statement implies that 
if there are $K+1$ distinct eigenvalues, then 
$\text{rank}(\bar{\Gammamat}_K )=K+1$.

Second, we prove Claim \ref{claim2} by showing that
\begin{equation} \label{Matrix_full_rank_1}
   \text{rank}((\text{diag}(\Phivec(N,R)\avec))^{-1} \Phivec(N,Q)) = Q+1
\end{equation} 
under the condition in Claim \ref{claim2}, where
$\Phivec(N,R)$ and $\Phivec(N,Q)$ are defined in \eqref{Gammamat_def}, and $Q+1\leq N$.
We assume that $[\Phivec(N,R)\avec]_n \neq 0$,  $\forall n=1,\dots,N$ and thus, $\left(\text{diag}(\Phivec(N,R)\avec)\right)$ is a non-singular matrix.
The multiplication of $\Phivec(N,Q)$ by an $N\times N$ non-singular matrix, $(\text{diag}(\Phivec(N,R)\avec))^{-1}$, in \eqref{Matrix_full_rank_1} implies that (see 0.4.6 in \cite{Horn_Johnson_book})
\begin{equation} \label{Matrix_full_rank_2}
 \text{rank}((\text{diag}(\Phivec(N,R)\avec))^{-1} \Phivec(N,Q)) = \text{rank}(\Phivec(N,Q)).\hspace{-0.07cm} 
\end{equation} 
Using the properties (see 0.9.11 \cite{Horn_Johnson_book}) of 
the Vandermonde matrix, $\Phivec(N,Q)$, we obtain that if there are $Q+1$ distinct eigenvalues, then $\text{rank}(\Phivec(N,Q))=Q+1$.
By using \eqref{Matrix_full_rank_2}, this statement implies that 
if there are $Q+1$ distinct eigenvalues, then 
$\text{rank}((\text{diag}(\Phivec(N,R)\avec))^{-1} \Phivec(N,Q)) = Q+1$.

\bibliographystyle{IEEEtran}

\begin{thebibliography}{10}
\providecommand{\url}[1]{#1}
\csname url@samestyle\endcsname
\providecommand{\newblock}{\relax}
\providecommand{\bibinfo}[2]{#2}
\providecommand{\BIBentrySTDinterwordspacing}{\spaceskip=0pt\relax}
\providecommand{\BIBentryALTinterwordstretchfactor}{4}
\providecommand{\BIBentryALTinterwordspacing}{\spaceskip=\fontdimen2\font plus
\BIBentryALTinterwordstretchfactor\fontdimen3\font minus
  \fontdimen4\font\relax}
\providecommand{\BIBforeignlanguage}[2]{{%
\expandafter\ifx\csname l@#1\endcsname\relax
\typeout{** WARNING: IEEEtran.bst: No hyphenation pattern has been}%
\typeout{** loaded for the language `#1'. Using the pattern for}%
\typeout{** the default language instead.}%
\else
\language=\csname l@#1\endcsname
\fi
#2}}
\providecommand{\BIBdecl}{\relax}
\BIBdecl

\bibitem{7544580}
W.~{Huang}, L.~{Goldsberry}, N.~F. {Wymbs}, S.~T. {Grafton}, D.~S. {Bassett},
  and A.~{Ribeiro}, ``Graph frequency analysis of brain signals,'' \emph{IEEE
  J. Sel. Topics Signal Processing}, vol.~10, no.~7, pp. 1189--1203, Oct. 2016.

\bibitem{1583238}
R.~{Olfati-Saber} and J.~S. {Shamma}, ``Consensus filters for sensor networks
  and distributed sensor fusion,'' in \emph{Proc. of CDC}, 2005, pp.
  6698--6703.

\bibitem{8347162}
A.~Ortega, P.~Frossard, J.~Kovačević, J.~M.~F. Moura, and P.~Vandergheynst,
  ``Graph signal processing: Overview, challenges, and applications,''
  \emph{Proc. IEEE}, vol. 106, no.~5, pp. 808--828, May 2018.

\bibitem{Shuman_Ortega_2013}
D.~I. Shuman, S.~K. Narang, P.~Frossard, A.~Ortega, and P.~Vandergheynst, ``The
  emerging field of signal processing on graphs: {E}xtending high-dimensional
  data analysis to networks and other irregular domains,'' \emph{IEEE Signal
  Processing Magazine}, vol.~30, no.~3, pp. 83--98, May 2013.

\bibitem{Isufi_Leus2017}
E.~{Isufi}, A.~{Loukas}, A.~{Simonetto}, and G.~{Leus}, ``Autoregressive moving
  average graph filtering,'' \emph{IEEE Trans. Signal Process.}, vol.~65,
  no.~2, pp. 274--288, 2017.

\bibitem{9244650}
Y.~{Tanaka}, Y.~C. {Eldar}, A.~{Ortega}, and G.~{Cheung}, ``Sampling signals on
  graphs: From theory to applications,'' \emph{IEEE Signal Processing
  Magazine}, vol.~37, no.~6, pp. 14--30, 2020.

\bibitem{6854325}
A.~{Anis}, A.~{Gadde}, and A.~{Ortega}, ``Towards a sampling theorem for
  signals on arbitrary graphs,'' in \emph{Proc. of ICASSP}, May 2014, pp.
  3864--3868.

\bibitem{7352352}
A.~G. {Marques}, S.~{Segarra}, G.~{Leus}, and A.~{Ribeiro}, ``Sampling of graph
  signals with successive local aggregations,'' \emph{IEEE Trans. Signal
  Process.}, vol.~64, no.~7, pp. 1832--1843, Apr. 2016.

\bibitem{9043719}
Y.~{Tanaka} and Y.~C. {Eldar}, ``Generalized sampling on graphs with subspace
  and smoothness priors,'' \emph{IEEE Trans. Signal Process.}, vol.~68, pp.
  2272--2286, 2020.

\bibitem{9343697}
Z.~{Xiao}, H.~{Fang}, and X.~{Wang}, ``Distributed nonlinear polynomial graph
  filter and its output graph spectrum: Filter analysis and design,''
  \emph{IEEE Trans. Signal Processing}, vol.~69, pp. 1--15, 2021.

\bibitem{8347160}
G.~B. {Giannakis}, Y.~{Shen}, and G.~V. {Karanikolas}, ``Topology
  identification and learning over graphs: Accounting for nonlinearities and
  dynamics,'' \emph{Proc. IEEE}, vol. 106, no.~5, pp. 787--807, 2018.

\bibitem{shen2016nonlinear}
Y.~{Shen}, G.~B. {Giannakis}, and B.~{Baingana}, ``Nonlinear structural vector
  autoregressive models with application to directed brain networks,''
  \emph{IEEE Trans. Signal Process.}, vol.~67, no.~20, pp. 5325--5339, 2019.

\bibitem{8496842}
Z.~{Xiao} and X.~{Wang}, ``Nonlinear polynomial graph filter for signal
  processing with irregular structures,'' \emph{IEEE Trans. Signal Processing},
  vol.~66, no.~23, pp. 6241--6251, 2018.

\bibitem{drayer2018detection}
E.~{Drayer} and T.~{Routtenberg}, ``Detection of false data injection attacks
  in smart grids based on graph signal processing,'' \emph{IEEE Systems
  Journal}, vol.~14, no.~2, pp. 1886--1896, 2020.

\bibitem{Grotas_2019}
S.~{Grotas}, Y.~{Yakoby}, I.~{Gera}, and T.~{Routtenberg}, ``Power systems
  topology and state estimation by graph blind source separation,'' \emph{IEEE
  Trans. Signal Processing}, vol.~67, no.~8, pp. 2036--2051, 2019.

\bibitem{shaked2021identification}
S.~Shaked and T.~Routtenberg, ``Identification of edge disconnections in
  networks based on graph filter outputs,'' \emph{IEEE Trans. Signal and
  Information Processing over Networks}, vol.~7, pp. 578--594, 2021.

\bibitem{bienstock2019strong}
D.~{Bienstock} and A.~{Verma}, ``Strong {NP}-hardness of {AC} power flows
  feasibility,'' \emph{Oper. Res. Lett.}, vol.~47, no.~6, pp. 494--501, 2019.

\bibitem{Abor}
A.~Abur and A.~Gomez-Exposito, \emph{Power System State Estimation: Theory and
  Implementation}.\hskip 1em plus 0.5em minus 0.4em\relax Marcel Dekker, 2004.

\bibitem{6891254}
E.~{Björnson}, J.~{Hoydis}, M.~{Kountouris}, and M.~{Debbah}, ``Massive {MIMO}
  systems with non-ideal hardware: Energy efficiency, estimation, and capacity
  limits,'' \emph{IEEE Trans. Information Theory}, vol.~60, no.~11, pp.
  7112--7139, 2014.

\bibitem{7931630}
Y.~{Li}, C.~{Tao}, G.~{Seco-Granados}, A.~{Mezghani}, A.~L. {Swindlehurst}, and
  L.~{Liu}, ``Channel estimation and performance analysis of one-bit massive
  {MIMO} systems,'' \emph{IEEE Trans. Signal Process.}, vol.~65, no.~15, pp.
  4075--4089, 2017.

\bibitem{berman2020resource}
I.~E. Berman and T.~Routtenberg, ``Resource allocation and dithering of
  {B}ayesian parameter estimation using mixed-resolution data,'' \emph{IEEE
  Trans. Signal Processing}, 2021.

\bibitem{Eldar_merhav}
Y.~C. {Eldar} and N.~{Merhav}, ``A competitive minimax approach to robust
  estimation of random parameters,'' \emph{IEEE Trans. Signal Process.},
  vol.~52, no.~7, pp. 1931--1946, 2004.

\bibitem{Edfors_Borjesson_1996}
O.~{Edfors}, M.~{Sandell}, J.~J. {van de Beek}, S.~K. {Wilson}, and P.~O.
  {Borjesson}, ``{OFDM} channel estimation by singular value decomposition,''
  in \emph{Proc. of Vehicular Technology Conference}, vol.~2, 1996, pp.
  923--927.

\bibitem{dd}
N.~{Geng}, X.~{Yuan}, and L.~{Ping}, ``Dual-diagonal {LMMSE} channel estimation
  for {OFDM} systems,'' \emph{IEEE Trans. Signal Process.}, vol.~60, no.~9, pp.
  4734--4746, 2012.

\bibitem{7032244}
S.~{Chen}, A.~{Sandryhaila}, J.~M.~F. {Moura}, and J.~{Kovacevic}, ``Signal
  denoising on graphs via graph filtering,'' in \emph{Proc. of GlobalSIP},
  2014, pp. 872--876.

\bibitem{ZHANG20083328}
F.~Zhang and E.~R. Hancock, ``Graph spectral image smoothing using the heat
  kernel,'' \emph{Pattern Recognition}, vol.~41, pp. 3328 -- 3342, 2008.

\bibitem{6778068}
S.~{Chen}, F.~{Cerda}, P.~{Rizzo}, J.~{Bielak}, J.~H. {Garrett}, and
  J.~{Kovačević}, ``Semi-supervised multiresolution classification using
  adaptive graph filtering with application to indirect bridge structural
  health monitoring,'' \emph{IEEE Trans. Signal Process.}, vol.~62, no.~11, pp.
  2879--2893, 2014.

\bibitem{6808520}
A.~{Sandryhaila} and J.~M.~F. {Moura}, ``Discrete signal processing on graphs:
  Frequency analysis,'' \emph{IEEE Trans. Signal Processing}, vol.~62, no.~12,
  pp. 3042--3054, June 2014.

\bibitem{sparse_paper}
J.~{Liu}, E.~{Isufi}, and G.~{Leus}, ``Filter design for autoregressive moving
  average graph filters,'' \emph{IEEE Trans. Signal Inf. Process. Netw.},
  vol.~5, no.~1, pp. 47--60, 2019.

\bibitem{7926424}
S.~{Segarra}, A.~G. {Marques}, and A.~{Ribeiro}, ``Optimal graph-filter design
  and applications to distributed linear network operators,'' \emph{IEEE Trans.
  Signal Process.}, vol.~65, no.~15, pp. 4117--4131, 2017.

\bibitem{7117446}
S.~{Chen}, A.~{Sandryhaila}, J.~M.~F. {Moura}, and J.~{Kovačević}, ``Signal
  recovery on graphs: Variation minimization,'' \emph{IEEE Trans. Signal
  Process.}, vol.~63, no.~17, pp. 4609--4624, Sept. 2015.

\bibitem{7891646}
N.~{Perraudin} and P.~{Vandergheynst}, ``Stationary signal processing on
  graphs,'' \emph{IEEE Trans. Signal Process.}, vol.~65, pp. 3462--3477, 2017.

\bibitem{Isufi_Ribeiro}
L.~Ruiz, F.~Gama, and A.~Ribeiro, ``Graph neural networks: Architectures,
  stability, and transferability,'' \emph{Proceedings of the IEEE}, vol. 109,
  no.~5, pp. 660--682, 2021.

\bibitem{gama2020graphs}
F.~Gama, E.~Isufi, G.~Leus, and A.~Ribeiro, ``Graphs, convolutions, and neural
  networks: From graph filters to graph neural networks,'' \emph{IEEE Signal
  Processing Magazine}, vol.~37, no.~6, pp. 128--138, 2020.

\bibitem{7763882}
J.~{Mei} and J.~M.~F. {Moura}, ``Signal processing on graphs: Causal modeling
  of unstructured data,'' \emph{IEEE Trans. Signal Processing}, vol.~65, no.~8,
  pp. 2077--2092, 2017.

\bibitem{Hua_Sayed_2020}
F.~{Hua}, R.~{Nassif}, C.~{Richard}, H.~{Wang}, and A.~H. {Sayed}, ``Online
  distributed learning over graphs with multitask graph-filter models,''
  \emph{IEEE Trans. Signal Inf. Process. Netw.}, vol.~6, pp. 63--77, 2020.

\bibitem{rey2021robust}
\BIBentryALTinterwordspacing
S.~Rey and A.~G. Marques, ``Robust graph-filter identification with graph
  denoising regularization,'' 2021. [Online]. Available:
  \url{https://arxiv.org/abs/2103.05976}
\BIBentrySTDinterwordspacing

\bibitem{confPaper}
A.~Kroizer, Y.~C. Eldar, and T.~Routtenberg, ``Modeling and recovery of graph
  signals and difference-based signals,'' in \emph{Proc. of GlobalSIP}, Nov.
  2019, pp. 1--5.

\bibitem{Newman_2010}
M.~Newman, \emph{Networks: An Introduction}.\hskip 1em plus 0.5em minus
  0.4em\relax New York, NY, USA: Oxford University Press, Inc., 2010.

\bibitem{6638849}
A.~Sandryhaila and J.~M.~F. Moura, ``Discrete signal processing on graphs:
  Graph filters,'' in \emph{Proc. of ICASSP}, 2013, pp. 6163--6166.

\bibitem{4072747}
F.~{Fouss}, A.~{Pirotte}, J.~{Renders}, and M.~{Saerens}, ``Random-walk
  computation of similarities between nodes of a graph with application to
  collaborative recommendation,'' \emph{IEEE Trans. Knowledge and Data
  Engineering}, vol.~19, no.~3, pp. 355--369, 2007.

\bibitem{4053117}
F.~{Fouss}, L.~{Yen}, A.~{Pirotte}, and M.~{Saerens}, ``An experimental
  investigation of graph kernels on a collaborative recommendation task,'' in
  \emph{Prod. of ICDM}, 2006, pp. 863--868.

\bibitem{StatisticalDigitalSignalProcessingHayes}
M.~H. Hayes, \emph{Statistical Digital Signal Processing and Modeling},
  1st~ed.\hskip 1em plus 0.5em minus 0.4em\relax USA: John Wiley \& Sons, Inc.,
  1996.

\bibitem{7001054}
X.~{Shi}, H.~{Feng}, M.~{Zhai}, T.~{Yang}, and B.~{Hu}, ``Infinite impulse
  response graph filters in wireless sensor networks,'' \emph{IEEE Signal
  Processing Letters}, vol.~22, no.~8, pp. 1113--1117, 2015.

\bibitem{7208894}
S.~{Chen}, R.~{Varma}, A.~{Sandryhaila}, and J.~{Kovačević}, ``Discrete
  signal processing on graphs: Sampling theory,'' \emph{IEEE Trans. Signal
  Processing}, vol.~63, no.~24, pp. 6510--6523, 2015.

\bibitem{routtenberg2020}
T.~Routtenberg, ``Non-{B}ayesian estimation framework for signal recovery on
  graphs,'' \emph{IEEE Trans. Signal Process.}, vol.~69, pp. 1169--1184, 2021.

\bibitem{dabush2021state}
\BIBentryALTinterwordspacing
L.~Dabush, A.~Kroizer, and T.~Routtenberg, ``State estimation in unobservable
  power systems via graph signal processing tools,'' 2021. [Online]. Available:
  \url{https://arxiv.org/abs/2106.02254}
\BIBentrySTDinterwordspacing

\bibitem{letting2017estimation}
L.~K. Letting, Y.~Hamam, and A.~M. Abu-Mahfouz, ``Estimation of water demand in
  water distribution systems using particle swarm optimization,'' \emph{Water},
  vol.~9, no.~8, p. 593, 2017.

\bibitem{alaoui2019computational}
\BIBentryALTinterwordspacing
A.~E. Alaoui and A.~Montanari, ``On the computational tractability of
  statistical estimation on amenable graphs,'' 2019. [Online]. Available:
  \url{https://arxiv.org/abs/1904.03313}
\BIBentrySTDinterwordspacing

\bibitem{Kayestimation}
S.~M. Kay, \emph{Fundamentals of statistical signal processing: Estimation
  Theory}.\hskip 1em plus 0.5em minus 0.4em\relax Englewood Cliffs (N.J.):
  Prentice Hall PTR, 1993, vol.~1.

\bibitem{6827237}
J.~{Serra} and M.~{N\'ajar}, ``Asymptotically optimal linear shrinkage of
  sample {LMMSE} and {MVDR} filters,'' \emph{IEEE Trans. Signal Process.},
  vol.~62, no.~14, pp. 3552--3564, 2014.

\bibitem{van2004optimum}
H.~L. Van~Trees, \emph{Optimum array processing: Part IV of detection,
  estimation, and modulation theory}.\hskip 1em plus 0.5em minus 0.4em\relax
  John Wiley \& Sons, 2004.

\bibitem{4914845}
F.~{Rubio} and X.~{Mestre}, ``Consistent reduced-rank {LMMSE} estimation with a
  limited number of samples per observation dimension,'' \emph{IEEE Trans.
  Signal Process.}, vol.~57, no.~8, pp. 2889--2902, 2009.

\bibitem{5484583}
Y.~{Chen}, A.~{Wiesel}, Y.~C. {Eldar}, and A.~O. {Hero}, ``Shrinkage algorithms
  for {MMSE} covariance estimation,'' \emph{IEEE Trans. Signal Process.},
  vol.~58, no.~10, pp. 5016--5029, 2010.

\bibitem{ledoit2004well}
O.~Ledoit and M.~Wolf, ``A well-conditioned estimator for large-dimensional
  covariance matrices,'' \emph{Journal of multivariate analysis}, vol.~88,
  no.~2, pp. 365--411, 2004.

\bibitem{SVD}
L.~Dieci and T.~Eirola, ``On smooth decompositions of matrices,'' \emph{SIAM J.
  on Matrix Anal. and Appl.}, vol.~20, no.~3, pp. 800--819, 1999.

\bibitem{Dong_Vandergheynst_2016}
X.~Dong, D.~Thanou, P.~Frossard, and P.~Vandergheynst, ``Learning {L}aplacian
  matrix in smooth graph signal representations,'' \emph{IEEE Trans. Signal
  Process.}, vol.~64, no.~23, pp. 6160--6173, Dec. 2016.

\bibitem{ramezani2019graph}
M.~Ramezani-Mayiami, M.~Hajimirsadeghi, K.~Skretting, R.~S. Blum, and H.~V.
  Poor, ``Graph topology learning and signal recovery via {B}ayesian
  inference.'' in \emph{DSW}, 2019, pp. 52--56.

\bibitem{Rama2020Anna}
R.~{Ramakrishna}, H.~T. {Wai}, and A.~{Scaglione}, ``A user guide to low-pass
  graph signal processing and its applications: Tools and applications,''
  \emph{IEEE Signal Processing Magazine}, vol.~37, no.~6, pp. 74--85, 2020.

\bibitem{Horn_Johnson_book}
R.~Horn and C.~R. Johnson, \emph{Matrix Analysis}, 2nd~ed.\hskip 1em plus 0.5em
  minus 0.4em\relax New York, NY: Cambridge University Press, 2012.

\bibitem{oellermann1991laplacian2}
B.~Mohar, ``The {L}aplacian spectrum of graphs, [in {G}raph {T}heory,
  {C}ombinatorics, and {A}pplications, eds: {Y}. {A}lavi, {G}. {C}hartrand,
  {O}. {O}llermann, and {A}. {S}chwenk,,'' pp. 871--898, 1991.

\bibitem{9206091}
F.~Gama, J.~Bruna, and A.~Ribeiro, ``Stability properties of graph neural
  networks,'' \emph{IEEE Trans. Signal Processing}, vol.~68, pp. 5680--5695,
  2020.

\bibitem{gao2021stability}
Z.~Gao, E.~Isufi, and A.~Ribeiro, ``Stability of graph convolutional neural
  networks to stochastic perturbations,'' \emph{Signal Processing}, p. 108216,
  2021.

\bibitem{kenlay2021interpretable}
\BIBentryALTinterwordspacing
H.~Kenlay, D.~Thanou, and X.~Dong, ``Interpretable stability bounds for
  spectral graph filters,'' 2021. [Online]. Available:
  \url{https://arxiv.org/abs/2103.05976}
\BIBentrySTDinterwordspacing

\bibitem{le2017approximate}
L.~Le~Magoarou, R.~Gribonval, and N.~Tremblay, ``Approximate fast graph
  {F}ourier transforms via multilayer sparse approximations,'' \emph{IEEE
  Trans. Signal and Information Processing over Networks}, vol.~4, no.~2, pp.
  407--420, 2017.

\bibitem{lu2019fast}
K.-S. Lu and A.~Ortega, ``Fast graph {F}ourier transforms based on graph
  symmetry and bipartition,'' \emph{IEEE Trans. Signal Processing}, vol.~67,
  no.~18, pp. 4855--4869, 2019.

\bibitem{iEEEdata}
\BIBentryALTinterwordspacing
``Power systems test case archive.'' [Online]. Available:
  \url{http://www.ee.washington.edu/research/pstca/}
\BIBentrySTDinterwordspacing

\end{thebibliography}

\end{document}